\documentclass[11pt]{article}
\usepackage[utf8]{inputenc}
\usepackage[dvipsnames]{xcolor}
\usepackage{authblk}
\usepackage{placeins}
\usepackage[dvipsnames]{xcolor}
\definecolor{DarkGreen}{rgb}{0.1,0.5,0.1}
\definecolor{DarkRed}{rgb}{0.5,0.1,0.1}
\definecolor{DarkBlue}{rgb}{0.1,0.1,0.5}

\usepackage{mleftright}

\usepackage[small]{caption}
\usepackage[pdftex]{hyperref}
\hypersetup{
    unicode=false,          
    pdftoolbar=true,        
    pdfmenubar=true,        
    pdffitwindow=false,      
    pdfnewwindow=true,      
    colorlinks=true,       
    linkcolor=DarkBlue,          
    citecolor=DarkGreen,        
    filecolor=DarkGreen,      
    urlcolor=DarkBlue,          
    pdftitle={},
    pdfauthor={},    
}
\usepackage{amssymb,amsmath,amsthm,amsfonts,enumitem, dsfont}
\usepackage{fullpage,comment,mathtools}
\usepackage{tikz}
\usepackage{pgfplots}
\pgfplotsset{compat=1.14}
\usetikzlibrary{arrows,decorations.pathmorphing,decorations.shapes,snakes,patterns,positioning}
\usepackage[ruled,linesnumbered]{algorithm2e}
\usepackage{thm-restate}
\mathtoolsset{showonlyrefs}

\newcommand{\cB}{\ensuremath{\mathcal{B}}}

\newcommand{\cP}{\ensuremath{\mathcal{P}}}

\newcommand{\cI}{\ensuremath{\mathcal{I}}}

\newcommand{\cT}{\ensuremath{\mathcal{T}}}

\newcommand{\R}{{\mathbb R}}
\newcommand{\Z}{{\mathbb Z}}

\newcommand{\F}{{\mathbb F}}

\newcommand{\N}{\ensuremath{\mathbb{N}}}

\newcommand{\PR}[1]{\mathrm{Pr}\left[ #1\right]}
\newcommand{\PROver}[2]{\mathrm{Pr}_{#1}\left[ #2\right]}
\renewcommand{\Pr}{\mathrm{Pr}}
\newcommand{\E}[1]{\mathbb{E}\left[ #1\right]}
\newcommand{\EE}{\mathbb{E}}
\newcommand{\Eover}[2]{\mathop{\mathbb{E}}_{#1}\left[ #2 \right]}

\newcommand{\VAR}[1]{\mathrm{Var}\left[ #1\right]}

\newcommand{\PRover}[2]{\mathop{\mathrm{Pr}}_{#1}\left[ #2\right]}

\newcommand{\floor}[1]{\left\lfloor {#1}\right\rfloor}

\newcommand{\inset}[1]{\left\{#1\right\}}

\newcommand{\inparen}[1]{\left(#1\right)}

\newcommand{\suchthat}{\,:\,}

\newcommand{\supp}{\mathrm{supp}}

\newcommand{\argmax}{\mathrm{argmax}}
\newcommand{\argmin}{\mathrm{argmin}}

\newcommand{\poly}{\mathrm{poly}}

\newcommand{\Sym}{\ensuremath{\operatorname{Sym}}}
\newcommand{\dis}[2]{#1_{\ensuremath{\operatorname{distinct}}}^{#2}}

\newcommand{\DKL}[3]{{D_{\mathrm{KL}}}_{#3}\left(#1\parallel #2\right)}

\newcommand{\TRC}{R_{\mathrm{RC}}}
\newcommand{\CRC}{C_{\mathrm{RC}}}

\newcommand{\TRLC}{R_{\mathrm{RLC}}}
\newcommand{\CRLC}{C_{\mathrm{RLC}}}

\newcommand{\histypes}{{\bar T}}

\newcommand\numberthis{\addtocounter{equation}{1}\tag{\theequation}}

\newcommand{\eps}{\varepsilon}
\renewcommand{\epsilon}{\varepsilon}

\definecolor{brightpink}{rgb}{1.0, 0.0, 0.5}
\definecolor{byzantine}{rgb}{0.74, 0.2, 0.64}
\definecolor{byzantium}{rgb}{0.44, 0.16, 0.39}

\DeclareMathOperator*\Uni{Uniform}

\theoremstyle{plain}
\declaretheorem[name=Theorem,numberwithin=section]{theorem}
\declaretheorem[name=Lemma,sibling=theorem]{lemma}
\newtheorem*{lemma*}{Lemma} 
\newtheorem*{theorem*}{Theorem} 
\newtheorem{definition}[theorem]{Definition}

\newtheorem{corollary}[theorem]{Corollary} 

\newtheorem{remark}[theorem]{Remark}
\newtheorem{claim}[theorem]{Claim}
\newtheorem{proposition}[theorem]{Proposition}

\newtheorem{fact}[theorem]{Fact}

\DeclareFontFamily{U}{mathx}{\hyphenchar\font45}
\DeclareFontShape{U}{mathx}{m}{n}{<-> mathx10}{}
\DeclareSymbolFont{mathx}{U}{mathx}{m}{n}
\DeclareMathAccent{\widebar}{0}{mathx}{"73}

\title{\textbf{Threshold rates for properties of random codes}\thanks{
		SS and MW are partially funded by NSF-CAREER grant CCF-1844628, NSF-BSF grant CCF-1814629, and a Sloan Research Fellowship. SS is partially supported by a Google Graduate Fellowship. VG, JM and NR are partially funded by NSF grants CCF-1563742 and CCF-1814603 and a Simons Investigator Award. NR is also partially supported by ERC H2020 grant No.74079 (ALGSTRONGCRYPTO).}}

\author[1]{Venkatesan Guruswami}
\author[1]{Jonathan Mosheiff} 
\author[2]{Nicolas Resch}
\author[3]{\\Shashwat Silas}
\author[3]{Mary Wootters}
\affil[1]{Carnegie Mellon University}
\affil[2]{Centrum Wiskunde en Informatica}
\affil[3]{Stanford University}
\date{}
\begin{document}
	\maketitle
	\thispagestyle{empty}
	
	\begin{abstract}
		Suppose that $\mathcal{P}$ is a property that may be satisfied by a random code $C \subset \Sigma^n$.  For example, for some $p \in (0,1)$, $\mathcal{P}$ might be the property that there exist three elements of $C$ that lie in some Hamming ball of radius $pn$.  We say that $R^*$ is the \em threshold rate \em for $\mathcal{P}$ if a random code of rate $R^* + \eps$ is very likely to satisfy $\mathcal{P}$, while a random code of rate $R^* - \eps$ is very unlikely to satisfy $\mathcal{P}$.
		While random codes are well-studied in coding theory, even the threshold rates for relatively simple properties like the one above are not well understood.
		
		\smallskip
		We characterize threshold rates for a rich class of properties.  These properties, like the example above, are defined by the inclusion of specific sets of codewords which are also suitably ``symmetric.'' For properties in this class, we show that the threshold rate is in fact \em equal \em to the lower bound that a simple first-moment calculation obtains. Our techniques not only pin down the threshold rate for the property $\mathcal{P}$ above, they give sharp bounds on the threshold rate for \emph{list-recovery} in several parameter regimes, as well as an efficient algorithm for estimating the threshold rates for list-recovery in general. 
	\end{abstract}
	
	\newpage
	
	\section{Introduction} 
	
Random codes are ubiquitous in the theory of error correcting codes: when thinking about the ``right’’ trade-offs for a particular problem, a coding theorist’s first instinct may be to try a random code. 
A \em random code \em here is simply a random set.  That is, let $C \subseteq \Sigma^n$ be chosen so that each $x \in \Sigma^n$ is included in $C$ with probability $|\Sigma|^{-n(1-R)}$ for some parameter $R$, which is called the (expected\footnote{Throughout, we refer to $R$ as the rate of the code, and drop the adjective ``expected.''}) \em rate \em of the code $C$.  
Random codes are used in the proofs of the Gilbert-Varshamov bound, Shannon’s channel coding theorem, and the list-decoding capacity theorem, to name just a few.  This success may lead to the intuition that random codes are ``easy’’ to analyze, and that the hard part is finding explicit constructions that match (or in rare cases, exceed) the parameters of random codes.  However, there is still much we do not know about random codes, especially if we want extremely precise answers.

In particular, the question of \em threshold rates\em, of broader interest in probability theory, is something that we do not understand well for random codes.
In more detail, suppose that $\mathcal{P}$ is a code property.  For example, perhaps $\mathcal{P}$ is the property that there is some pair of codewords $c^{(1)}, c^{(2)} \in C$ that both lie in some Hamming ball of radius $pn$.
Or perhaps $\mathcal{P}$ is the property that there are three codewords $c^{(1)}, c^{(2)}, c^{(3)} \in C$ that lie in such a Hamming ball.  A value $R^* \in (0,1)$ is a \em threshold rate \em for $\mathcal{P}$ if a random code of rate $R^* + \eps$ is very likely to satisfy $\mathcal{P}$, but a random code of rate $R^* - \eps$ is very unlikely to satisfy $\mathcal{C}$.
For the first example above, about pairs of codewords, the property in question is just the property of the code having \emph{minimum distance} less than $2pn$, and this is not too hard to understand. 
%
%
However, already for the second example above---called \emph{list-of-two decoding}---the threshold rate was not known.

\subsection{Contributions}
In this paper, we characterize threshold rates for a rich class of natural properties of random codes.  We apply our characterization to obtain threshold rates for list-of-two decoding, as well as to properties like \em list-decoding \em and \em perfect hashing codes, \em and more generally to \em list-recovery. \em 
We outline our contributions below.

\paragraph{A characterization of the threshold rate $R^*$ for symmetric properties.}  Suppose that $\mathcal{P}$ is a property defined by the inclusion of certain ``bad’’ sets.  For example, the list-of-two decoding property described above is defined by the inclusion of three codewords that lie in a radius-$pn$ Hamming ball.  
 For such properties that are also ``symmetric enough,'' our main technical result, Theorem~\ref{thm:threshold-intro}, characterizes the threshold rate $R^*$.  Moreover, we show that this threshold rate is exactly the same as the lower bound that one obtains from a simple first-moment calculation!  This is in contrast to recent work of \cite{MRRSW} for random \em linear \em codes, which shows that the corresponding first-moment calculation is not the correct answer in that setting.
 
Part of our contribution is formalizing the correct notion of ``symmetric enough.'' As we describe in the technical overview in Section~\ref{sec:techoverview}, this definition turns out to be fairly subtle.  Moreover, we give an example in Appendix~\ref{app:example} that shows that this definition is necessary: there are natural properties that do not meet this requirement, for which the simple first-moment calculation is \em not \em the correct threshold rate.


\paragraph{Estimates of $R^*$ for list-recovery.} 
We give precise estimates of the threshold rate $R^*$ for \em list-recovery. \em 
We say that a code $C \subseteq \Sigma^n$ is $(p, \ell, L)$-list-recoverable if for all sets $K_i \subseteq \Sigma$ (for $1 \leq i \leq n$) with $|K_i| \leq \ell$, 
\[ |\{ c \in C \,:\, \Pr_{i \sim [n]}[c_i \not\in K_i] \leq p \}| < L. \]
List-recovery is a useful primitive in list-decoding, algorithm design, and pseudorandomness~(see, e.g., \cite{RudraW18,guruswami2009unbalanced,Vadhan12}).  In particular, it generalizes the list-of-two decoding example above (when $\ell=1$ and $L=3$), as well as other interesting properties, such as list-decoding and perfect hashing codes, discussed below.

Our characterization allows us to estimate or even exactly compute the threshold rate for $(p, \ell, L)$-list-recovery in a wide variety of parameter regimes.  To demonstrate this, we include several results along these lines.  First, in Section~\ref{sec:ThresholdComputation} (Corollary~\ref{cor:thresh}), we give estimates that are quite sharp when $\frac{q\log L}{L}$ is small.  In Section~\ref{sec:hash} (Lemma~\ref{lem:zero}), we give an exact formula for the case $p=0$, which is relevant for perfect hashing codes.  In Section~\ref{sec:list-of-2} (Theorem~\ref{thm:list-of-2}(I)), we give an exact formula for the case that $L=3$ and $\ell=1$, relevant for list-of-two decoding.  Moreover, in Section~\ref{sec:efficient} (Corollary~\ref{cor:algorithm}) we use our characterization to develop an efficient algorithm to compute the threshold rate up to an additive error of $\epsilon > 0$; our algorithm runs in time  $O_p( L^q + \poly(q, L, \log(1/\eps))$.

\paragraph{List-of-two decoding and a separation between random codes and random linear codes.} 
We obtain new results for list-of-two decoding, the example discussed above.  
List-of-two decoding is a special case of \em list-decoding, \em which itself the special case of list-recovery where $\ell=1$.  We say that a code is $(p,L)$-list-decodable if there is no Hamming ball of radius $pn$ containing $L$ codewords; list-of-two decoding is the special case of $L=3$.\footnote{It is called list-of-\emph{two} decoding, even though $L$ is \emph{three}, because any Hamming ball contains at most \emph{two} codewords.}
We show in Section~\ref{sec:list-of-2} (Theorem~\ref{thm:list-of-2}) that the threshold rate for this question, for random binary codes, is $R^* = 1 - \frac{1 - h_2(3p) + 3p \log_2 3 }{3}$.  That is, above this rate, a random binary code is very likely to have three codewords contained in a radius $pn$ ball, while below this rate, the code most likely avoids all such triples. 

This result is interesting for two reasons.  First, it demonstrates that our techniques are refined enough to pin down the threshold rate in this parameter regime.  Second, the particular value of $R^*$ is interesting because it is \emph{different} than the corresponding threshold rate for random \emph{linear} codes.  A \em random linear code \em over $\mathbb{F}_q$ of rate $R$ is a random linear subspace of $\mathbb{F}^n_q$, of dimension $Rn$.  The list-decodability of random linear codes has been extensively studied, and it is known (e.g.,~\cite{ZyablovP81,GHK11}) that the $(p,L)$-list-decoding threshold rate for both random linear codes and random codes is $1 - h_q(p)$, for sufficiently large list sizes $L$.\footnote{Here, $h_q(x) = x\log_q(q-1) - x \log_q(x) - (1-x)\log_q(1-x)$ is the $q$-ary entropy.}
On the other hand, it is well-known that the distance of a random linear code (corresponding to the $L=2$ case) is better than the distance of a completely random code.  Our results show that this difference continues to hold for $L=3$ (list-of-two decoding).  We compute the threshold for random codes in Section~\ref{sec:list-of-2}, and
in Appendix~\ref{app:list-of-2-rlc}, we show how to use the techniques of \cite{MRRSW,GLMRSW20} to establish the list-of-two decoding threshold for random linear codes.
The $L=3$ case is interesting on its own, and moreover this result is a proof-of-concept to show that these techniques could perhaps be used to pin down the difference between random codes and random linear codes for larger (but still small) values of $L$.

\paragraph{Limitations of random codes for perfect hashing.}  
Another special case of list-recovery is \em perfect hashing codes. \em 
Suppose that $|\Sigma| = q$.  A code $C \subseteq \Sigma^n$ is said to be a $q$-hash code if, for any set of $q$ distinct codewords $c^{(1)}, c^{(2)}, \ldots, c^{(q)} \in C$, there is at least one $i \in [n]$ so that $\{c^{(1)}_i, c^{(2)}_i, \ldots, c^{(q)}_i\} = \Sigma$; that is, if the set of symbols that appear in position $i$ are all distinct.  Thus, $C$ is a $q$-hash code if and only if it is $(0,q-1,q)$-list-recoverable.  As the name suggests, $q$-hash codes have applications in constructing small perfect hash families, and it is a classical question to determine the largest rate possible for a $q$-hash code.\footnote{A $q$-hash code naturally gives rise to a perfect hash family: suppose that $C$ is a universe of items, and define a hash function $h_i:C \to \Sigma$ given by $h_i(c) = c_i$.  Then the property of being a $q$-hash code is equivalent to the property that, for any set of $q$ items in the universe, there exists some hash function $h_i$ for $1 \leq i \leq n$ that maps each item to a different value.} 

A simple random coding argument shows that a random code of rate $R = \frac{1}{q} \log_q\frac{1 }{1 - q!/q^q} - o(1)$ is a $q$-hash code with high probability~\cite{FK84,korner86}.  However, it is still an area of active research to do significantly better than this bound for any $q$.  It is known that $R < \frac{q!}{q^{q-1}}$ for any $q$-hash code~\cite{FK84,GR19}, and for large $q$, there is a gap of a multiplicative factor of about $q^2$ between these upper and lower bounds.  K\"orner and Matron gave a construction that beats the random bound for $q=3$~\cite{KM88}, and recently Xing and Yuan gave a construction that beats the random bound for infinitely many $q$'s~\cite{XY19}.  

One might have hoped that a random code might in fact do better than the straightforward probabilistic argument (which follows from a union bound).  Unfortunately, our results show that this is not the case.  In Corollary~\ref{cor:hash}, we use our characterization to pin down the threshold rate for perfect hashing codes, and show that, for random codes, the threshold rate is in fact $R^* = \frac{1}{q} \log_q\frac{1 }{1 - q!/q^q}$.

\paragraph{A broader view.}
Taking a broader view, threshold phenomena in other combinatorial domains, notably random graphs and Boolean functions, have been the subject of extensive study at least since Erd\H{o}s and R\'{e}nyi's seminal work \cite{ER59}. Some of the deeper results in this field (e.g.\ \cite{Fri99}), deal simultaneously with a wide class of properties, rather than a specific one. Other works, such as the recent \cite{FrankstonKNP19}, are general enough to cover not only multiple properties, but also multiple domains.
Our work (as with the work of \cite{MRRSW} on random linear codes, discussed below) is not as general as these, but we are able to get more precise results.  
%
It would be interesting to find a  general framework that connects threshold phenomena in a variety of random code models, with analogues from random graphs and other natural combinatorial structures.

\subsection{Technical Overview}\label{sec:techoverview}
As mentioned above, we study properties defined by the inclusion of bad subsets.  We organize bad subsets of size $b$ into matrices $B \in \Sigma^{n \times b}$, interpreting the columns of $B$ as the elements of the set.  We write ``$B \subseteq C$'' to mean that the columns of $B$ are all contained in the code $C$.

As a running example---and also our motivating example---consider list recovery, defined above. 
The property $\cP$ of \em not \em being $(p,\ell,L)$-list-recoverable is defined by the inclusion of ``bad'' matrices $B \in \Sigma^{n \times L}$ so that for some sets $K_1, \ldots, K_n \subset \Sigma$ of size at most $\ell$, $\Pr_{i \sim [n]}[B_{ij} \notin K_i ] \leq p$ for each $j \in [L]$. Moreover we require the columns of $B$ to be distinct. 

\paragraph{Analyzing a property as a union of types.}
Following the approach of~\cite{MRRSW} for random linear codes, we group the bad matrices into \em types \em based on their row distributions.  That is, for a bad matrix $B \in \Sigma^{n \times b}$, let $\tau$ denote the row distribution
\[ \tau(v) = \frac{ |\{i \in [n] : B_{i,\star} = v\}| }{n}, \] 
where $B_{i,\star}$ denotes the $i$'th row of $B$.  We say that $B$ has \em type \em $\tau$.  Consider the set $\mathcal{B}$ of all of the matrices of type $\tau$; equivalently, $\mathcal{B}$ is the set of matrices obtained by permuting the rows of $B$.

For example, if $B \subset \{\alpha, \beta, \gamma\}^{n\times 3}$ is given by
\begin{center}
    \begin{tikzpicture}
    \node {\begin{minipage}{2cm}
\[ B = \begin{bmatrix} \alpha & \alpha & \beta \\ \alpha & \alpha& \beta \\  & \vdots & \\ \alpha & \alpha& \beta \\ \gamma & \beta & \beta \\ & \vdots &  \\\gamma & \beta & \beta \\\gamma & \beta & \beta \end{bmatrix} \]
\end{minipage}};
    \draw [decorate,decoration={brace,amplitude=10pt},xshift=4pt,yshift=0pt] (1.5,1.7) -- (1.5,-.2)node [black,midway,xshift=9pt,anchor=west] {$n/2$};
        \draw [decorate,decoration={brace,amplitude=10pt},xshift=4pt,yshift=0pt] (1.5,-.25) -- (1.5,-2.1)node [black,midway,xshift=9pt,anchor=west] {$n/2$};
    \end{tikzpicture}
    \end{center}
then $B$ is bad for $(0, 2, 3)$-list-recovery, the row distribution $\tau$ of $B$ is given by
    \[ \tau( (\alpha, \alpha, \beta) ) = \tau( (\gamma, \beta, \beta) ) = \frac{1}{2}, \] 
and the set $\mathcal{B}$ consists of all $n \times 3$ matrices with half the rows $(\alpha, \alpha, \beta)$ and the other half $(\gamma, \beta, \beta)$.  

We note that possible types $\tau$ depend on $n$, because of divisibility constraints.  For simplicity, let us ignore these restrictions for now (we will deal with them later), and suppose that a single type $\tau$ can appear for all $n$.

\paragraph{First-moment bound and main theorem.}
We can use a simple first-moment approach to give a lower bound on the threshold rate. 
In more detail, the probability that a particular $B$ is contained in $C$ is $q^{-nb(1-R)}$, assuming that $B$ has $b$ distinct columns.
Using the fact that $|\mathcal{B}| \approx q^{H_q(\tau)\cdot n}$, where $H_q(\tau)$ is the base-$q$ entropy of $\tau$ (see Section~\ref{sec:Preliminaries}), and applying a union bound over all $B \in \mathcal{B}$, we see that the probability that any $B \in \mathcal{B}$ is contained in $C$ is at most
\[ q^{nb(H_q(\tau) - (1-R))}. \]
Thus, if $R \leq 1 - \frac{H_q(\tau)}{b} - \eps$ for some small $\eps > 0$, it is very unlikely that $\tau$ will be represented in $C$.

Now suppose that our collection of bad sets, which define the property $\mathcal{P}$, is closed under row permutations.  This means that $\mathcal{P}$ can be represented as a collection $T$ of types $\tau$; note that the size of $T$ is polynomial in $n$.  Union bounding over all of these types, the computation above shows that a random code $C$ of rate $R < 1 - \max_{\tau \in T} \frac{H_q(\tau)}{b} - \eps$ will, with high probability, not satisfy $\mathcal{P}$.

The question is, could the rate be larger?  Might it be the case that $\mathcal{P}$ still not satisfied (with high probability) by a random code of rate $R$ significantly larger than $1 - \max_{\tau}H_q(\tau)/b$?  In \cite{MRRSW}, it was shown that the answer for random \em linear \em codes is ``yes.''  If $\mathcal{P}$ exhibits certain linear structure, then it may be possible that a higher rate random linear code still does not satisfy $\mathcal{P}$ with high probability.  One may conjecture that something similar holds for random codes.

Our main technical result, Theorem~\ref{thm:threshold}, is that, for random codes, for sufficiently symmetric properties, the answer to this question is ``no.''  That is, the simple calculation above \em does \em give the right answer for random codes!   

\begin{theorem}[Informal; see Theorem~\ref{thm:threshold} for the formal version]\label{thm:threshold-intro}
Let $\mathcal{P}$ be a ``symmetric'' property defined by the inclusion of a type among the types in $T$.
Let 
\[ R^* = 1 - \frac{\max_{\tau \in T} H_q(\tau)}{b} \]
Then for all $\eps > 0$, a random code of rate $R \geq R^* + \eps$ satisfies $\mathcal{P}$ with probability $1 - o(1)$, while a random code of rate $R^* - \eps$ satisfies $\mathcal{P}$ with probability $o(1)$.
\end{theorem}

\paragraph{Sketch of proof: second moment method.}
Below, we sketch the proof of Theorem~\ref{thm:threshold-intro}, and explain what the assumption of ``symmetry'' means.
As noted above, it is straightforward to show that the threshold rate $R^*$ is at least $1 - \max_{\tau\in T} \frac{H_q(\tau)}{b}$, so the challenge is to show that it is not larger.  The proof of Theorem~\ref{thm:threshold-intro} uses the second-moment method to show that for any \em histogram type \em $\tau$ (we discuss histogram types more below), a random code $C$ of rate $1 - H_q(\tau)/b + \eps$ is very likely to contain some matrix $B$ with type $\tau$.  Thus, the threshold rate is at most $1 - \max_{\tau} H_q(\tau)/b$, where the maximum is over all histogram types $\tau$ that appear in $T$.  Our eventual definition of ``symmetric'' will guarantee that it is legitimate to restrict our attention to histogram types.

\paragraph{Histogram types and the meaning of ``symmetry.''}
In order to apply the second moment method, we bound the variance of $\sum_{B \sim \tau} \mathbf{1}[B \subset C]$, where the sum is over all matrices $B$ of type $\tau$.  This turns out to be possible when $\tau$ has the following symmetry property: for any $u \in \Sigma^b$, and for any permutation $\pi:[b]\to[b]$, it holds that $\tau(u) = \tau(\pi(u))$, where $\pi(u)$ denotes the corresponding coordinate permutation of $u$.  We call such a type $\tau$ a \em histogram-type \em (Definition~\ref{def:histogramType}) because the probability of a particular vector $u$ under $\tau$ depends only on the histogram of $u$.  

A first attempt to formulate a definition of ``symmetry'' for Theorem~\ref{thm:threshold-intro} is thus to require $\mathcal{P}$ to be defined by histogram types.  This results in a true statement, but unfortunately it is too restrictive: it is not hard to see that, for example, the property of not being list-decodable contains types $\tau$ that are not histogram types.  Fortunately, for the logic above to go through, it is enough to show that $T$ contains a type $\tau$ that is \em both \em a maximum entropy distribution in $T$, and is also a histogram type.  Thus, the assumption of ``symmetry'' we will use is that $T$, the collection of types represented in the property $\mathcal{P}$, forms a convex set.  Then, using the fact that $\mathcal{P}$ is defined by the inclusion of bad sets (which do not care about the order of the columns in the corresponding matrices), we can always find a maximum entropy histogram type by ``symmetrizing’’ and taking a convex combination of column permutations of some maximum entropy type $\tau$.  

One might wonder if this symmetrization step (and the resulting assumption about convexity) is necessary.  In fact, it is.  In Appendix~\ref{app:example}, we give an example of a property $\mathcal{P}$ that is given by the union of bad types, but which is not closed under such ``symmetrization''.   For this property, the threshold rate turns out to be larger than the rate predicted by Theorem~\ref{thm:threshold-intro}.

\paragraph{Taking a limit as $n \to \infty$.}
There is one more challenge to consider, which is that in the description above, we have ignored the fact that we would like our characterization to work for a sequence of values of $n$.  However, a type $\tau$ only works for certain values of $n$ due to divisibility restrictions.  To get around this, we work instead with a sequence of types $\tau_n$ which tend to $\tau$.
This leads us to our final definition of ``symmetric’’ (Definition~\ref{def:SymmetricProperty}).  Suppose that $\mathcal{P}$ is a property defined by the inclusion of size-$b$ bad sets.  Then for each $n$, there is some collection $T^n$ of bad types $\tau_n$, each of which is a distribution on $\Sigma^b$. 
We say that $\mathcal{P}$ is \em symmetric \em if the sets $T^n$ approach some convex set $T$ as $n$ goes to infinity.  The logic above then goes through to give Theorem~\ref{thm:threshold-intro}.

\paragraph{Applications to list-recovery.}
Finally, in order to apply Theorem~\ref{thm:threshold-intro}, we need to understand the maximum entropy distribution $\tau$ for our property $\mathcal{P}$.  We do this for the property $\mathcal{P}$ of not being $(p,\ell,L)$-list-recoverable in a variety of parameter regimes in Sections~\ref{sec:ThresholdComputation}, \ref{sec:hash} and \ref{sec:list-of-2}, and along the way obtain our results about list-of-two decoding and perfect hashing codes.
Finally, in Section~\ref{sec:efficient}, we use our framework to develop an algorithm to efficiently calculate the threshold rate for $(p,\ell,L)$-list-recovery.

\subsection{Related Work}

Below, we briefly survey work in a few categories that relates to our results.

\paragraph{Limitations of random codes.}
Random codes have been studied in coding theory since its inception, as they are used to prove fundamental existential results.  However, typically these constructions show that a random code is ``good'' (e.g., list-decodable, or a $q$-hash code) with high probability, not that a random code of slightly higher rate is \em not \em good.  Our work establishes a threshold rate $R^*$ for random codes, meaning that in particular it establishes both a positive and a negative result.  
There are a few works that establish limitations for random codes. For example, it is known that random codes of rate $1 - h_q(p) - \eps$ require list sizes of at least $\Omega(1/\eps)$~\cite{GuruswamiN14,LiW18}, which is a larger lower bound than that which is known for general codes. 

\paragraph{Sharp thresholds for random linear codes.} Our work takes inspiration from the techniques of \cite{MRRSW}, which recently developed sharp thresholds for random \em linear \em codes.  The starting points for that work and our work are similar: in particular, that work also classified bad sets into types based on their row distribution.  However, as discussed above, the situation for random linear codes is different than the situation for random codes, because for random linear codes looking at the entropy-maximizing distribution $\tau$ in a property is not enough: one has to look at suitable projections of these distributions $\tau$.  One of the contributions of this work is showing that for random codes, there are no such complications, at least for suitably symmetric properties.  The entropy-maximizing distribution in a property directly determines the threshold rate.

\paragraph{List-of-two decoding.} 
As mentioned above, list-of-two decoding is a special case of $(p,L)$-list-decoding when $L=3$.
For larger $L$, a classical result known as the list-decoding capacity theorem says that a 
random code $C \subseteq \Sigma^n$ of rate $R \leq 1 - h_{q}(p) - \eps$, where $q = |\Sigma|$, is $(p, O(1/\eps))$ list-decodable with high probability; while if $R \geq 1 - h_{q}(p) + \eps$, then no code of rate $R$ is $(p,L)$ list-decodable for any $L = o(2^n)$.  Thus, for large enough list sizes, the threshold rate for list-decoding is $R^* = 1 - h_q(p)$.

While the list-size $L$ is often viewed as a growing parameter in the list-decoding literature, it is also interesting to consider fixed values of $L$. A natural first step is to consider list-of-two decoding, which in our notation is $(p,3)$-list-decodability.  We focus on the case where $q=2$.  Many prior works have studied the best rate (in terms of $p$) at which list-of-two decoding is possible.
First of all, Blinovsky~\cite{blinovsky1986bounds} demonstrated that in order to have positive rate $(p,3)$-list-decodable codes for infinitely many block lengths $n$, it is necessary and sufficient to have $p<1/4$. 
Moreover, \cite{blinovsky1986bounds} gave bounds on this best rate; these bounds were further improved in \cite{ashikhmin2000new}.  Results concerning list-of-2 decoding in the zero-rate regime (that is, where the rate tends to zero as $n$ tends to infinity) have also been obtained in, e.g.~\cite{Elias91, alon2018list}.  The works summarized here have focused on the \em best \em possible rate for list-of-two-decoding.  In contrast, our work focuses on random codes in particular, with the goal of being able to pin down the threshold rate very precisely.  Our result shows that random codes are \em not \em the best for list-of-two decoding.  

\subsection{Future directions and open questions}

We have given a characterization of the threshold rate for a rich class of properties to be satisfied by a random code, and we have focused on list-recovery as an example property.  We hope that this framework, and its potential extensions, will be useful more broadly.  We mention a few open questions and directions for future research.

\paragraph{Extending the zero-rate regime?} Due to the list-recovery capacity theorem, it is known that when $p \geq 1-\ell/q$, there are no positive rate $(p,\ell,L)$-list-recoverable codes over an alphabet of size $q$, for any (subexponential) list size $L$. However, for small values of $L$, it is natural to expect that even for values of $p$ slightly less than $1-\ell/q$, there are still no positive-rate $(p,\ell,L)$-list-recoverable codes. For example, if we want a positive rate $(p,2,3)$-list-recoverable code over an alphabet of size $5$, must we have $p\leq 0.599$? We show that this is the case for ``most'' codes: more precisely, for any integers $\ell,L$ there exists a $p^* < 1-\ell/q$ such that a positive rate random code is $(p,\ell,L)$-list-recoverable with high probability if and only if $p<p^*$ (see Remark~\ref{rem:zero-rate-regime}). It would be interesting to determine whether or not this holds for \emph{all} codes: namely, can we prove that a $(p,\ell,L)$-list-recoverable code must have zero rate if $p\geq p^*$ for the same value of $p^*$? 

\paragraph{Other properties of random codes?} We have focused on list-recovery and its special cases as examples of symmetric properties.  What other examples are there?  Could we extend our work beyond symmetric properties?  As the example in Appendix~\ref{app:example} shows, extending our results beyond symmetric properties would require a more complicated expression than the one in Theorem~\ref{thm:threshold-intro}.

\paragraph{Other ensembles of codes?}  Our framework draws inspiration from the work \cite{MRRSW} that develops a similar framework for random linear codes.  What other random ensembles of codes are amenable to such a framework?  Is there some meta-framework that would encompass more of these?  A potential starting point would be to study \emph{pseudolinear} codes~\cite{guruswami2004list,guruswami2001linear} which in some sense interpolate between uniformly random codes (as we study in this work) and random linear codes (as addressed by \cite{MRRSW,GLMRSW20}).  It would be also be interesting to study random (and random linear) codes in the rank-metric~\cite{ding,GuruswamiR18}. Random subspaces of Euclidean space are of significant interest in many areas including compressed sensing, dimensionality reduction, and Euclidean sections, and it would be interesting to investigate to what extent, if at all, the frameworks in \cite{MRRSW,GLMRSW20} or this work might apply to these settings.   

\paragraph{Sharp thresholds for list-of-$(L-1)$ decoding?}  We used our framework to compute the threshold for list-of-two decoding (e.g., $L=3$), and we believe that our techniques could be used for pinning down the threshold rate, in terms of $L$, for any $L$.  While this is known asymptotically for large $L$~\cite{GuruswamiN14,LiW18}, the precise value remains open for constant $L > 3$.

\subsection{Organization}
In Section~\ref{sec:Preliminaries}, we introduce notation, and also set up the definitions we need about types, thresholds, properties, and various notions of symmetry.  We also introduce (non-)list-recoverability as a property, and prove in Corollary~\ref{cor:ListRecoveryIsAnAmazingProperty} that it is symmetric.

In Section~\ref{sec:ThersholdCharacterization}, we state and prove Theorem~\ref{thm:threshold}, the formal version of the characterization theorem (Theorem~\ref{thm:threshold-intro} above).  
At the end of Section~\ref{sec:ThersholdCharacterization}, we begin to apply Theorem~\ref{thm:threshold} to list-recovery, and in particular define several notions we will need to analyze list recovery in the subsequent sections.

In the remaining sections, we specialize to list-recovery.
In Section~\ref{sec:ThresholdComputation}, we develop bounds on the threshold rate $R^*$ for list-recovery that are tight when $(q\log L)/L$ is small. In Section~\ref{sec:hash}, we compute the threshold rate $R^*$ exactly for zero-error list-recovery (that is, when $p=0$), and use this to compute the threshold rate for perfect hashing.  In Section~\ref{sec:list-of-2}, we compute the threshold rate $R^*$ for list-of-two decoding (e.g., list-recovery when $\ell=1$ and $L=3$), and use this to quantify the gap between random codes and random linear codes for list-of-two decoding.  Finally, in Section~\ref{sec:efficient}, we give an efficient algorithm to compute the threshold rate.
	
	\section{Preliminaries}\label{sec:Preliminaries}
	
First, we fix some basic notation.  Throughout, we consider codes $C\subseteq \Sigma^n$ of block length $n$ over an alphabet $\Sigma$, where $|\Sigma|=q$.   When we use $\log(x)$ without an explicit base, we mean $\log_2(x)$.  We use $H_q$ to denote the base-$q$ entropy: for a distribution $\tau$, 
\[ H_q(\tau) := -\sum_x \tau(x) \log_q(\tau(x)).\]
When $q$ is clear from context, we will use $H(\tau)$ to denote $H_q(\tau)$. 
If $u$ is a random variable distributed according to $\tau$, then we abuse notation slightly and define $H(u) := H(\tau)$.
We use $h_q(x) := x \log_q(q-1) - x \log_q(x) - (1 -x) \log_q(1-x)$ to denote the $q$-ary entropy of $x \in (0,1)$.  Again, when $q$ is clear from context we will use $h(x)$ to denote $h_q(x)$.

For a vector $x \in \Sigma^k$ and $I \subseteq [k]$, we use $x_I$ to refer to the vector $(x_i)_{i \in I} \in \Sigma^I$.  Given a vector $u \in \Sigma^k$ and a permutation $\pi:[k] \to [k]$, we let $\pi(u) \in \Sigma^k$ denote the corresponding coordinate permutation of $u$.

Given distributions $\tau,\mu$ on the same finite set, we define their $\ell_\infty$-distance by 
\[
    d_\infty(\tau,\mu) := \max_{x}|\tau(x)-\mu(x)| \ .
\]
Given a set of distributions $T$, we define the $\ell_\infty$ distance from $\mu$ to $T$ by 
\[
    d_\infty(\mu,T) := \inf_{\tau \in T}d_\infty(\mu,\tau) \ .
\]






\subsection{Basic notions}

As mentioned in the introduction, we will organize our ``bad'' sets into matrices.  We formalize this with the following two definitions.

\begin{definition}[Matrices with distinct columns]
Let $\dis{\Sigma}{n\times b}$ denote the collection of all matrices $B\in \Sigma^{n\times b}$ such that each column of $B$ is distinct.
\end{definition}

\begin{definition}[Subsets as matrices]
Let $C\subseteq \Sigma^n$ be a code, and let $B \in \Sigma^{n\times b}$ be a matrix. We write $B\subseteq C$ to mean that each column of $B$ is an element of $C$. If $A\subseteq \Sigma^n$, let $\cB_A\subseteq \Sigma^{n\times |A|}$ denote the collection all matrices $B\in \dis{\Sigma}{n\times |A|}$ such that the columns of $B$ are the elements of $A$.
\end{definition}

For completeness, we reiterate our definition of a random code from the introduction.
\begin{definition}[Random code]
Let $\Sigma$ be a finite set with $q:=|\Sigma|\ge 2$. For $n\in \N$ and $R\in [0,1]$, let $\CRC^n(R)$ denote an \emph{expected-rate $R$ random code (over the alphabet $\Sigma$)} $C\subseteq \Sigma^n$. Namely, for each $x\in \F_q^n$ we have $\PR{x\in C} = q^{-n(1-R)}$, and these events are independent over all $x$.
\end{definition}

We record a useful fact about random codes, which is the probability that any particular matrix $B$ is contained in one.
\begin{fact}[Probability that a random code contains a matrix]\label{fact:ProbabilityContainedInMatrix}
Let $B\in \Sigma^{n\times b}$. Then, $$\Pr\left[B\subseteq \CRC^n(R) \right] = q^{-n(1-R)t},$$
where $t$ is the number of distinct columns in $B$.
\end{fact}

We study (noisy) list-recovery, which generalizes both the list-decoding and perfect hashing examples mentioned in the introduction. 
We repeat the definition, so that we may formally define a ``bad'' matrix for list-recovery.

\begin{definition}[Noisy list-recovery]
Let $p\in [0,1]$, $1\le \ell\le q$, and $L\in N$. Say that a matrix $B\in \dis{\Sigma}{L\times n}$ is \emph{$(p,\ell,L)$-bad} for $(p,\ell,L)$-list-recovery if there exist sets $K_i\subseteq \Sigma$ ($1\le i\le n$), each of size $\ell$, such that for every $1\le j\le L$,
\begin{equation}\label{eq:badMatrix}
    \PROver{i\sim [n]}{B_{i,j}\notin K_i} \le p.
\end{equation}

A code $C\subseteq \Sigma^n$ is \emph{$(p,\ell,L)$-list-recoverable} if it does not contain a $(p,\ell,L)$-bad matrix.
\end{definition}

\subsection{Monotone-increasing properties and thresholds}

We study the threshold rate $R^*$ for random codes to satisfy certain properties.  This was discussed informally in the introduction and the definitions below formalize what ``threshold rate'' means.
\begin{definition}[Monotone-increasing property]
A code property $\cP$ is \emph{monotone-increasing} if given a code $C$ satisfying $\cP$, it holds that every code $C'$ such that $C\subseteq C'$ also satisfies $\cP$.
\end{definition}

For example, the property of being \em not \em $(p,\ell,L)$-list-recoverable (that is, the property of containing a $(p,\ell,L)$-bad matrix) is a monotone-increasing property.

\begin{definition}[Minimal-set]
Let $P_n$ be a monotone-increasing property of length-$n$ codes. A set $A\subseteq \Sigma^n$ is a \emph{minimal element} of $P_n$ if $A$ satisfies $P_n$ but no strict subset of $A$ satisfies $P_n$. The \emph{minimal set} for $P_n$ is the collection of matrices $$\bigcup_{A\text{ is a minimal element of }  P_n} \cB_{A}.$$
\end{definition}

For example, the minimal set for the property $P_n$ of being \em not \em $(p,\ell,L)$-list-recoverable is the set of $(p, \ell, L)$-bad matrices.

Note that a code satisfies $P_n$ if and only if it contains some matrix belonging to the minimal set of $P_n$. If $\cP$ is a monotone-increasing property of codes, we denote its associated \emph{threshold rate} by $\TRC^n(\cP)$. This is defined as
\begin{align*}
\sup\left\{R\in [0,1] : \Pr\left[\CRC^n(R)\text{ satisfies } \cP\right] \le \frac 12\right\}
\end{align*}

if there is such an $R$, and $0$ otherwise. 
\begin{remark}
If $\cP$ is monotone-increasing then the function $\Pr\left[\CRC^n(R) \text{satisfies }\cP\right]$ is monotone-increasing in $R$. This can be proved by a standard coupling argument, akin to \cite[Thm.\ 2.1]{Bollobas01}.
\end{remark}

\begin{definition}[Sharpness for random codes]
A monotone-increasing property $\cP$ is \emph{sharp for random codes} if 
$$\lim_{n\to \infty}\Pr\left[\CRC^n\left(\TRC^n(\cP)-\eps\right)\text{ satisfies } \cP\right] = 0$$
and
$$\lim_{n\to \infty}\Pr\left[\CRC^n\left(\TRC^n(\cP)+\eps\right)\text{ satisfies } \cP\right] = 1$$
for every $\eps > 0$.
\end{definition}

\subsection{Local and row-symmetric properties}

As discussed in the introduction, we study properties that can be written as a union of ``types,'' where each type corresponds to a row distribution $\tau$ of a matrix $M$.  The following definitions make this notion precise.

\begin{definition}[Row-permutation invariant collection of matrices]
A collection of matrices $\cB\subseteq \Sigma^{n\times b}$ is \emph{row-permutation invariant} if, given a matrix $B\in \cB$, every row permutation of $B$ (that is, a matrix resulting from applying the same coordinate permutation to each column of $B$) also belongs to $\cB$.
\end{definition}

\begin{definition}[Local and row-symmetric properties]
Let $\cP = \{P_n\}_{n\in \N}$ be a monotone-increasing property, and let $M_n$ denote the minimal set of $P_n$.
\begin{itemize}
    \item If there exists some $b\in \N$ such that $M_n\subseteq \Sigma^{n\times b}$ for every $n$, we say that $\cP$ is \emph{$b$-local}.
    \item If every $M_n$ is row-permutation invariant, we say that $\cP$ is \emph{row-symmetric}. 
\end{itemize}
\end{definition}

\begin{remark}
Every monotone-increasing property is trivially \emph{column-symmetric}, in the sense that permuting the columns of a matrix in $M_n$ results in another matrix in $M_n$. This naturally reflects the fact that containment of a matrix does not depend on the ordering of the columns, and follows immediately from the definition of a minimal set. 
\end{remark}

Let $B \in \Sigma^{n\times b}$, and consider the collection $\cB$ of all row-permutations of $B$. Let $\tau_B$ denote the \emph{row-distribution} of $B$. That is, $\tau$ is the probability distribution, over $\Sigma^b$, of the row $B_{i,\star}$, where $i$ is sampled uniformly from $[n]$. Observe that every matrix in $\cB$ has the same row-distribution as $B$. Moreover, $\cB$ can be characterized as the set of all matrices with the row distribution $\tau_B$. These observations motivate the following definitions.

\begin{definition}[Type of a matrix]
Let $B\in \Sigma^{n\times b}$. We define its \emph{type} $\tau_B$ as the distribution of a uniformly random row of $B$. That is, $\tau_B$ is the distribution over $\Sigma^b$, such that
$$\tau_B(x) = \frac{|\{i\in [n]\mid B_i=x\}|}n$$
for every $x\in \Sigma^b$. 
Let $$\cT_b^n = \{\tau_B\mid B\in \dis{\Sigma}{n\times b}\}$$ denote the set of all possible types of $n\times b$ matrices with distinct columns.
Given $\tau \in \cT_b^n$, we denote
$$M_\tau = \{B\in \Sigma^{n\times b}\mid \tau_B=\tau\}.$$
\end{definition}
\begin{remark}
The type of a matrix $B\in \Sigma^{n\times b}$ determines whether $B\in \dis{\Sigma}{n\times b}$. Therefore, for $\tau\in \cT_b^n$,
$$M_\tau = \{B\in \dis{\Sigma}{n\times b}\mid \tau_B=\tau\}.$$
\end{remark}
The following fact now follows from the above discussion.
\begin{fact}[Decomposition of a row-permutation invariant collection]\label{fact:TypeDecomposition}
Let $\cB\subseteq \Sigma^{n\times b}$ be a row-permutation invariant collection. Then, there exists a set of types $T\subseteq \cT_{n,b}$ such that $$\cB = \bigcup_{\tau\in T} M_\tau.$$
\end{fact}

Note that a type in $\cT_b^n$ is defined by the number of occurrences of each of $|\Sigma^b|$ possible rows, in a matrix consisting of $n$ rows. In particular, each row occurs between $0$ and $n$ times. Thus,
\begin{equation}\label{eq:NumberOfTypesBound}
|\cT^n_{b}| \le (n+1)^{|\Sigma^b|} = (n+1)^{q^b}.
\end{equation}
Crucially for our purposes, this upper bound is polynomial in $n$. 

\subsection{Symmetric properties and convex approximations}

\begin{definition}
 Let $\cT_b$ denote the simplex of all probability distributions over $\Sigma^b$.
\end{definition}

It is generally more convenient to work in $\cT_b$ rather than $\cT_b^n$, since the former is continuous, while the latter is discrete and involves certain divisibility conditions. This motivates the following definition.

\begin{definition}[Permutation-closed type sets]
A set $T\subseteq \cT_b$ is called \emph{permutation-closed} if for every $\tau \in T$ and every permutation $\pi:[b]\to[b]$, the distribution of $\pi(u)$ (where $u\sim \tau$) also belongs to $T$. 
\end{definition}

\begin{definition}[Approximating sets of types]\label{def:Approximation}
Fix $b\in \N$. Let $\{T^n\}_{n\in \N}$ be a sequence of sets of types, such that $T^n \subseteq \cT_b^n$. A (topologically) closed and permutation-closed set $T\subseteq \cT_b$ is an \emph{approximation} for $\{T^n\}_{n\in \N}$ if $T^n\subseteq T$ for every $n$, and 
\begin{equation}\label{eq:GoodApproximation}
    \lim_{n\to \infty} \max_{\tau \in T} d_\infty(\tau, T^n) = 0.
\end{equation}
\end{definition}

\begin{definition}[Symmetric property and convex approximation]\label{def:SymmetricProperty}
Let $\cP = \{P_n\}_{n\in \N}$ be a $b$-local, row-symmetric, monotone-increasing property. Due to Fact~\ref{fact:TypeDecomposition}, for every $n$ there exists a set $T_n\subseteq \cT_{n,b}$ such that the minimal set of $P_n$ is $\bigcup_{\tau\in T_n} M_\tau$. If the sequence $\{T_n\}_{n\in N}$ has a convex approximation $T$, we say that $T$ is a \emph{convex approximation} for $\cP$. In this case, we say that $\cP$ is \emph{symmetric}.
\end{definition}

\subsection{Non-list-recoverability as a property}\label{sec:noisylistrecisgood}
Our motivating property is that of being \em not \em list-recoverable.  In this section, we show that non-$(p,\ell,L)$-list-recoverability is a symmetric property, and we define the convex set $T_{p,\ell,L}$ that is a convex approximation for it.

Fix $p\in [0,1]$, $1\le \ell\le q$ and $L\in \N$. Let $\cP = \left(P_n\right)_{n\in \N}$ denote the property of being \emph{not} $(p,\ell,L)$-list-recoverable. That is, a code $C\subseteq \Sigma^n$ satisfies $P_n$ if it contains a $(p,\ell,L)$-bad matrix. We now show that $\cP$ is a symmetric property.

Clearly, $\cP$ is monotone-increasing, and its minimal set is exactly the set of $(p,\ell,L)$-bad matrices, which we denote $M_n\subseteq \dis{\Sigma}{n\times L}$. It follows immediately that $\cP$ is $L$-local. Furthermore since the left-hand side of \eqref{eq:badMatrix} is invariant to row-permutations of $B$, the collection $M_n$ is row-permutation invariant, and so $\cP$ is row-symmetric. 

Fact \ref{fact:TypeDecomposition} says that we can write $M_n = \bigcup_{\tau\in T_{p,\ell,L}^n} M_\tau$ for some $T_{p,\ell,L}^n\subseteq \cT_L^n$. Indeed, \eqref{eq:badMatrix} yields the following description of $\cT_{p,\ell,L}^n$: A type $\tau\in \cT_{L}^n$ belongs to $T_{p,\ell,L}^n$ if and only if there exists a distribution $\rho$ over $\Sigma^{L} \times {\binom{\Sigma}\ell}$ such that, given $(u,K)\sim \rho$, the following holds:
\begin{enumerate}
    \item\label{enum:rhoMarginalIsTau} The distribution of $u$ is $\tau$.
    \item\label{enum:rhoSmallRadius} For every $1\le j\le L$, it holds that $\PR{u_j\notin K}\le p$.
    \item\label{enum:rhoDivisibility} $n\cdot \rho((u,K))\in\N$ for every $u$ and $K$. 
\end{enumerate}
To see this, let $\rho$ be the joint distribution $(B_i,K_i)$ for $i$ uniformly sampled from $[n]$, where $B$ and $K$ are as in \eqref{eq:badMatrix}. Note that $\rho$ must satisfy the three conditions above. In the other direction, it is not hard to see that any such distribution $\rho$ as above gives rise to a matrix of type $\tau$, satisfying \eqref{eq:badMatrix}.

We next construct a convex approximation for $\cP$. Let $T_{p,\ell,L}$ denote the set of all types $\tau\in \cT_L$ for which there exists a distribution $\rho$ satisfying Conditions \ref{enum:rhoMarginalIsTau} and \ref{enum:rhoSmallRadius}, but not necessarily Condition~\ref{enum:rhoDivisibility}: 
\begin{definition}
Let $1\le \ell\le 1$, $L\in \N$ and $0\le p\le 1$. Let $\tau$ be a distribution over $\Sigma^L$. We say that $\tau$ belongs to the set $T_{p,\ell,L}$ if there exists a distribution $\rho$ over $\Sigma^L\times \binom{\Sigma}L$ such that:
\begin{enumerate}\label{def:TpellL}
    \item \label{enum:rhoMarginalIsTauForT}
    If $(u,K)\sim \rho$ then the vector $u$ is $\tau$-distributed.
    \item \label{enum:rhoSmallRadiusForT}
    For every $1\le j\le L$ it holds that
    \begin{equation}\label{eq:rhoSmallRadiusInequality}
    \PROver{(u,K)\sim \rho}{u_j\notin K}\le p.
\end{equation}
\end{enumerate}
\end{definition}

Clearly, $T_{p,\ell,L}^n\subseteq T_{p,\ell,L}$ for all $n \in \mathbb{N}$. It is also immediate to verify that $T_{p,\ell,L}$ is permutation-closed. 

\begin{lemma}
The set $T_{p,\ell,L}$ is convex.
\end{lemma}
\begin{proof}
Let $\tau_0, \tau_1\in T_{p,\ell,L}$. Let $t\in [0,1]$ and let $\tau_t$ denote the mixture distribution $(1-t)\tau_0 + t\tau_1$. Let $\rho_0$ and $\rho_1$ be distributions over $\Sigma^L\times \binom{\Sigma}\ell$, satisfying Conditions \ref{enum:rhoMarginalIsTauForT} and \ref{enum:rhoSmallRadiusForT} for $\tau_0$ and $\tau_1$, respectively. Let $\rho_t$ be the mixture distribution $(1-t)\rho_0 + t\rho_1$. It is straightforward to verify that $\rho_t$ satisfies Conditions \ref{enum:rhoMarginalIsTauForT} and \ref{enum:rhoSmallRadiusForT} with respect to $\tau_t$. Hence, $\tau_t\in T_{p,\ell,L}$.
\end{proof}

The following lemma, proven in Appendix \ref{app:ProofsForCharacterization} shows that $T_{p,\ell,L}$ satisfies \eqref{eq:GoodApproximation}. Namely, every type in $T_{p,\ell,L}$ can be realized with low error as a type from $T_{p,\ell,L}^n$, for large enough $n$.
\begin{restatable}{lemma}{GoodApproximationForNonRecoverability}\label{lem:GoodApproximationForNonRecoverability}
  $$\lim_{n\to \infty} \sup_{\tau \in T_{p,\ell,L}} d_\infty(\tau, T_{p,\ell,L}^n) = 0.$$
\end{restatable}

We record the results of this section in the following corollary.
\begin{corollary}\label{cor:ListRecoveryIsAnAmazingProperty}
Being \emph{not} $(p,\ell,L)$-list-recoverable is a symmetric property. Furthermore, $T_{p,\ell,L}$ is a convex approximation for this property.
\end{corollary}

	\section{Characterization theorem}\label{sec:ThersholdCharacterization}
	

In this section, we prove our main characterization theorem, Theorem~\ref{thm:threshold-intro}, which is formally stated below as Theorem~\ref{thm:threshold}.  Before stating and proving the theorem, we record a few useful lemmas.  

\begin{lemma}[{\cite[Lemma~2.2]{csiszar2004information}}]\label{lem:MatricesInTypeBound}
Let $\tau \in \cT_b^n$. 
Then,
$$ q^{H(\tau)n}\cdot n^{-O_{q,b}(1)}\le |M_\tau|\le q^{H(\tau)n}.$$
\end{lemma}

\begin{lemma}\label{lem:SetSizeUpperBound}
Let $M\subseteq \Sigma^{n\times b}$. Then, 
$$|M|\le (n+1)^{q^b}\cdot q^{n\cdot \max_{B\in M}H(\tau_B)}.$$
\end{lemma}
\begin{proof}
Let $T = \{\tau_B\mid B\in M\}$. Note that $$M\subseteq \bigcup_{\tau\in T} M_\tau.$$ Thus,
\begin{align*}|M| &\le  \sum_{\tau \in T} |M_\tau| \\& \le |T|\cdot \max_{\tau\in T}|M_\tau|\\& \le |\cT_{n,b}|\cdot \max_{\tau\in T}|M_\tau|.\end{align*}
The claim follows from \eqref{eq:NumberOfTypesBound} and Lemma \ref{lem:MatricesInTypeBound}.
\end{proof}

We say that a type is a \emph{histogram type} if it is indifferent to the ordering of a given vector's entries, and thus, only cares about the histogram of the vector. Formally, we make the following definition.

\begin{definition}[Histogram type]\label{def:histogramType}
A type $\tau\in T_b$ is called a \emph{histogram-type} if $\tau(u) = \tau(\pi(u))$ for every $u\in \Sigma^b$ and every permutation $\pi:[b]\to[b]$.
\end{definition}

\begin{lemma}\label{lem:SymmetricHasHighEntropy}
Let $T\subseteq \cT_b$ be a closed, permutation-closed, convex, set of types. Then there exists a histogram type $\tau\in T$ such that $H(\tau) = \max_{\tau'\in T}H(\tau')$.
\end{lemma}
\begin{proof}
Since $T$ is closed and bounded, it is compact. Thus, there is some $\tau'\in T$ such that $H(\tau')$ is maximal. Given a permutation $\pi:[b]\to [b]$, let $\pi(\tau')$ denote the distribution of the vector $\pi(u)$, where $u\sim \tau$. Let $$\tau = \frac{\sum_{\pi\in \Sym_b} \pi(\tau')}{b!}.$$ Since $T$ is permutation-closed and convex, $\tau\in T$. By concavity of entropy,
\begin{align*} H(\tau) &\ge \frac{\sum_{\pi\in \Sym_b}H(\pi(\tau'))}{b!} \\&= \frac{\sum_{\pi\in \Sym_k}H(\tau')}{b!} \\&= H(\tau').\end{align*}
Thus, $\tau$ has maximum entropy in $T$, and is clearly a histogram-type.
\end{proof}

The following technical lemma facilitates our use of an approximation for a set of types. 
\begin{lemma}\label{lem:entropyDifference}
Let $\tau,\tau' \in T_b$ such that $d_\infty(\tau,\tau') \le \eps.$ Then,
$$\left|H_{u\sim \tau}(u \mid u_I)-H_{u\sim\tau'}(u \mid u_I)\right|\le O_{b,q}\left(\eps\cdot\log \frac 1\eps\right)$$
for any $I\subseteq [b]$.
\end{lemma}

\begin{proof}
Given $x\in \Sigma^I$ and a distribution $\tau$ over $\Sigma^b$, write $\tau_I(x)=\PRover{u\sim \tau}{u_I=x}$. By our assumption, 
\begin{align}\label{eq:taubartauDiffBound}
    |\tau_I(x) - \tau_I'(x)| &= \left|\sum_{\substack{y\in \Sigma^b\\y_I=x}}(\tau_I(y)-\tau_I'(y))\right| \\&\le \sum_{\substack{y\in \Sigma^b\\y_I=x}}\left|\tau_I(y)-\tau_I'(y)\right| \\&\le|\Sigma^{[b]\setminus I}| \cdot\eps = q^{b-|I|}\cdot \eps
\end{align}
for all $I\subseteq [b]$, $x\in \Sigma^I$. 

We also need the following fact \cite[Eq. (4)]{Zhang07}: Let $\theta$ and $\theta'$ be two probability distributions on a set of $N$ elements, such that their $\ell_\infty$ distance is at most $\delta$. Then,
\begin{equation}\label{eq:dLinftyEntropyBound}
H(\theta)-H(\theta') \le \log_q2\cdot h_2(2N\delta) + 2N\delta\log_q N.
\end{equation}

Now,
\begin {align*}
&\left|H_{u\sim \tau}(u \mid u_I)-H_{u\sim \tau'}(u \mid u_I)\right| \\&= \left|\sum_{x\in \Sigma^I} \tau_I(x) H_{u\sim \tau}(u\mid u_I=x) \right. \\&\quad - \biggl. \tau_I'(x) H_{u\sim \tau'}(u\mid u_I=x) \biggr| \\
&\le\sum_{x\in \Sigma^I} |\tau_I(x) H_{u\sim \tau}(u\mid u_I=x) \\& \quad - \tau_I'(x) H_{u\sim \tau'}(u\mid u_I=x) | \numberthis\label{eq:Hdiff}
\end{align*}
We turn to bounding each term of this sum. Let $x \in \Sigma^I$ and assume without loss of generality that
$\tau_I(x)\le \tau_I'(x)$. By the triangle inequality,
\begin{align*}
&|\tau_I(x) H_{u\sim \tau}(u\mid u_I=x)  - \tau_I'(x) H_{u\sim \tau'}(u\mid u_I=x)|\\ &\le
|\tau_I(x)-\tau_I'(x)|\cdot H_{u\sim \tau'}(u\mid u_I=x) \\& \quad + \tau_I(x)\cdot|H_{u\sim \tau}(u\mid u_I=x) \\& \quad- H_{u\sim \tau'}(u\mid u_I=x)|.\numberthis\label{eq:HDiffTermSplit}
\end{align*}

Due to the fact that entropy can only decrease on conditioning, we can say that $H_{u\sim \tau}(u\mid u_I=x)\le H_{u\sim \tau}(u)\le b$. So, 
\begin{align}\label{eq:HDiffTerm1}
    &\left|\left(\tau_I(x)-\tau_I'(x)\right) H_{u\sim \tau'}(u\mid u_I=x)\right| \\& \le b\cdot q^{b-|I|}\cdot \eps
\end{align}
due to \eqref{eq:taubartauDiffBound}.

Let $\theta$ (resp. $\theta'$) denote the distribution of $u\sim \tau$ (resp. $u\sim \tau'$) conditioned on $u_I=x$. We bound the $L_\infty$ distance of $\theta$ and $\theta'$. For $u\in \Sigma^b$ such that $u_I=x$, we have
\begin{align*} &\left|\theta(u)-\theta'(u)\right| 
= \left|\frac{\tau(u)}{\tau_I(x)}-\frac{\tau'(u)}{\tau_I'(x)}\right| \\&\le \left|\frac{\tau(u)}{\tau_I(x)}-\frac{\tau(u)}{\tau_I'(x)}\right|+\left|\frac{\tau(u)}{\tau_I'(x)}-\frac{\tau'(u)}{\tau_I'(x)}\right| 
\\
&= \tau(u)\cdot\left|\frac{\tau_I'(x)-\tau_I(x)}{\tau_I(x)\cdot\tau_I'(x)}\right|+\frac{1}{\tau_I'(x)}\left|\tau(u)-\tau'(u)\right|\\
&\le \tau(u)\cdot\frac{q^{b-|I|}\cdot \eps}{\tau_I(x)\cdot\tau_I'(x)}+\frac{\eps}{\tau_I'(x)}\\ 
&\le \frac{q^{b-|I|}+1}{\tau_I'(x)}\cdot \eps,
\end{align*}
where the inequalities follow respectively from \eqref{eq:taubartauDiffBound}, and from the fact that $\tau(u)\le \tau_I(x)$. 
Hence, 
\begin{align*}
&\tau_I(x)\left|H_{u\sim \tau}(u\mid u_I=x) - H_{u\sim \tau'}(u\mid u_I=x)\right| \\&= \tau_I(x)\left|H(\theta) - H(\theta')\right|  \\
&\le \tau_I(x)  \left( \log_q 2\cdot h\left(2\frac{q^b+q^{|I|}}{\tau_I'(x)}\cdot \eps\right) + 2q^{|I|} |I| \eps\right)\\
&\le \tau_I'(x)  \left( \log_q 2\cdot h\left(2\frac{q^b+q^{|I|}}{\tau_I'(x)}\cdot \eps\right) + 2q^{|I|} |I| \eps\right)\\
&\le \tau_I'(x)  \left( \log_q 2\cdot h\left(2\frac{2q^b\cdot \eps}{\tau_I'(x)}\right) + 2q^b b \eps\right)\\
&\le O_{b,q}\left(\eps \log_q \frac 1\eps\right) \numberthis\label{eq:HDiffTerm2}
\end{align*}
due to \eqref{eq:dLinftyEntropyBound}. The lemma follows from \eqref{eq:Hdiff}, \eqref{eq:HDiffTermSplit}, \eqref{eq:HDiffTerm1} and \eqref{eq:HDiffTerm2}.
\end{proof}

We now prove that every monotone-increasing, local and row-symmetric property with a convex approximation is sharp for random codes. Furthermore, we identify the threshold rate as the maximal entropy in the approximating set.

\begin{theorem}\label{thm:threshold}
Fix $b\in \N$. Let $\cP = \{P_n\}_{n\in \N}$ be a symmetric property with locality parameter $b$, and let $T$ be a convex approximation for $\cP$. Denote $R^* = 1-\frac{\max_{\tau\in T}H(\tau)}{b}$.
Fix $\eps>0$ and let $R\in [0,1]$. The following now holds.
\begin{enumerate}
    \item If $R\le  R^*-\eps$ then $$\lim_{n\to \infty} \PR{\CRC^n(R) \text{ satisfies } \cP} = 0.$$
    \item If $R\geq  R^*+\eps$ then $$\lim_{n\to \infty} \PR{\CRC^n(R) \text{ satisfies } \cP} = 1.$$
\end{enumerate}
\end{theorem}
\begin{proof}
For $b\in \N$ and a matrix $B\in \dis{\Sigma}{b\times n}$, let $X_B$ be an indicator variable for the event that $B\in \CRC^n(R)$. For a set $M\subseteq \dis{\Sigma}{b\times n}$, let $X_M=\sum_{B\in M} X_B$.
By Fact \ref{fact:ProbabilityContainedInMatrix},
\begin{equation}\label{eq:E(X)}
\E{X_M} = |M|\cdot q^{-n(1-R)b}.
\end{equation}

Let $M_n$ denote the minimal set for $P_n$ and let $T_n = \{\tau_B\mid B\in M_n\}$.

The first statement now follows from Markov's inequality, \eqref{eq:E(X)}, and Lemma \ref{lem:SetSizeUpperBound}:
\begin{align*}
&\PR{C\text{ satisfies }\cP} \\&=  \PR{\exists B\in M_n~~B\subseteq \CRC^n(R)} \\
&\le \PR{X_M\ge 1}\\
&\le \E{X_M}\\
&= |M|\cdot q^{-n(1-R)b} \\
&\le (n+1)^{q^b}\cdot q^{n\cdot \max_{\tau\in T_n}H(\tau)}\cdot q^{-n(1-R)b}\\
&\le  (n+1)^{q^b}\cdot q^{n\cdot \max_{\tau\in T}H(\tau)}\cdot q^{-n(1-R)b} \\
&\le (n+1)^{q^b}\cdot q^{-nb\eps} \le e^{-\Omega(n)}.
\end{align*}
Above, we used the fact that $T_n\subseteq T$.

For the second statement, let $\tau\in T$ have maximum entropy. By definition~\ref{def:Approximation}, $T$ is closed and permutation-closed, in addition to being convex. Consequently, due to Lemma~\ref{lem:SymmetricHasHighEntropy}, we may assume that $\tau$ is a histogram-type. Let $\tau_n \in T_n$ such that $d_\infty (\tau,\tau_n) = o_{n\to \infty}(1)$. Our plan is to use a second-moment argument to show that $\CRC^n(R)$ likely contains a matrix of type $\tau_n$. 

By \eqref{eq:E(X)} and Lemma \ref{lem:MatricesInTypeBound},
\begin{align*}\E{X_{M_{\tau_n}}} &= |M_{\tau_n}|q^{-n(1-R)b} \\&\ge q^{(H({\tau_n})-(1-R)b)n+o(n)} \\&\ge q^{(H(\tau)-(1-R)b)n+o(n)}\end{align*}

We turn to bounding the variance of $X_{M_{\tau_n}}$. Fact~\ref{fact:ProbabilityContainedInMatrix} yields
\begin{align*}
&\VAR{X_{M_{\tau_n}}} \\&= \sum_{B,B'\in M_{\tau_n}}(\PR{X_B=X_{B'}=1} \\& \quad - \PR{X_B=1}\PR{X_{B'}=1}) \\
&=  \sum_{B,B'\in M_{\tau_n}} \left(q^{-n(1-R)(2b-\alpha(B,B'))}- q^{-2n(1-R)b}\right)\\
&\le \sum_{\substack{B,B'\in M_{\tau_n}\\ \alpha(B,B')\ge 1}} q^{-n(1-R)(2b-\alpha(B,B'))}
\end{align*}
where $\alpha(B,B')$ is the number of columns in $B'$ that also appear in $B$.

In order to bound this sum, we need an estimate on the number of pairs $B,B'$ with a given $\alpha(B,B')$. For $0\le r\le b$, let $$W_r = \{(B, B')\mid B,B'\in M_{\tau_n}\text{ and } \alpha(B,B')=r\}$$ 
and denote $S_r = \{\tau_{B\|B'}\mid (B,B')\in W_r\}$.
Here, $B \| B'$ is the $n\times 2b$ matrix whose first (resp. last) $b$ columns are $B$ (resp. $B'$).
By Lemma \ref{lem:SetSizeUpperBound}, 
$$|W_r|\le (n+1)^{2q^b}\cdot q^{n\max_{\nu\in S_r}H(\nu)}.$$
Let $(B,B')\in W_r$ and let $\nu = \tau_{B\|B'}$. Assume without loss of generality that the first $r$ columns of $B$ are identical to the first $r$ columns of $B'$. Let $u\sim \nu$. Note that, since $B,B'\in M_{\tau_n}$, the random variables $u_{[b]}$ and $u_{[2b]\setminus[b]}$ are both ${\tau_n}$-distributed. Hence,
\begin{align*}
H(\nu) &= H(u) = H(u_{[2b]\setminus[b]}) + H(u_{[b]}\mid u_{[2b]\setminus[b]}) \\&=  H({\tau_n}) + H(u_{[b]}\mid u_{[2b]\setminus[b]}) \\
&\le H({\tau_n}) + H(u_{[b]}\mid u_{[r]}) \\&= H({\tau_n}) + H(u_{[b]\setminus[r]}\mid u_{[r]}).
\end{align*}
Note that the inequality follows from the observation that $H(u_{[b]}\mid u_{[2b]\setminus[b]}) \le H(u_{[b]}\mid u_{[r]})$ since entropy can only decrease on conditioning. 

Lemma \ref{lem:entropyDifference} yields
\begin{align*}
    H(u_{[b]\setminus[r]}\mid u_{[r]}) &\le H_{v\sim \tau}(v_{[b]\setminus[r]}\mid v_{[r]}) + o(1)\\
    &=\sum_{i=r+1}^b H_{v\sim \tau}(v_i\mid v_{[i-1]}) + o(1)\\
    &= \sum_{i=r+1}^b H_{v\sim \tau}(v_b\mid v_{[i-1]}) + o(1),
\end{align*}
where the last equality is due to $\tau$ being a histogram-type. Writing 
\begin{equation}\label{eq:f(r)def}
f(r) = \sum_{i=r+1}^b H_{v\sim \tau}(v_b\mid v_{[i-1]}),
\end{equation}
we conclude that
$$H(\nu) \le f(r)+H(\tau)+o(1),$$
so that 
$$|W_r| \le q^{\left(f(r)+H(\tau)\right)n+o(n)},$$
and
\begin{align*}
&\VAR{X_{M_{\tau_n}}} \\&\le \sum_{r=1}^{b}|W_r| \cdot q^{-n(1-R)(2b-r)} \\&\le \sum_{r=1}^b q^{\left(f(r)+H(\tau)-(1-R)(2b-r)\right)n+o(n)}\\
&\le \max_{1\le r\le b} q^{\left(f(r)+H(\tau)-(1-R)(2b-r)\right)n+o(n)}\\
\end{align*}

By Chebyshev's inequality, 
\begin{align}\label{eq:Var(X)/E(X)2Bound}
    &\PR{X_{M_{\tau_n}} = 0} \\&\le \frac{\VAR{X_{M_{\tau_n}}}}{\E{X_{M_{\tau_n}}}^2} \\&\le \max_{1\le r\le b} q^{(f(r)-H(\tau)+r(1-R))n + o_{b,q}(n)}.
\end{align}

We claim that $(f(r))_{r=0}^b$ is a convex sequence. Indeed, 
\begin{align*}&f(r-1) + f(r+1) - 2f(r) \\&= H_{v\sim \tau}(v_b\mid v_{[r-1]}) - H_{v\sim \tau}(v_b\mid v_{[r]}) \\&\ge 0.\end{align*}
Therefore, the maximum in the right-hand side of \eqref{eq:Var(X)/E(X)2Bound} is achieved either by $r=1$ or $r=b$. In the former case, note that
\begin{align*}
f(1) &= \sum_{i=2}^b H_{v\sim \tau}(v_b\mid v_{[i-1]}) \\&= \sum_{i=2}^b H_{v\sim \tau}(v_i\mid v_{[i-1]}) = H_{v\sim \tau}(v\mid v_1) \
\\&\le H(\tau) - H_{v\sim \tau}(v_1) \le H(\tau)\cdot \frac{b-1}b.
\end{align*}
In the last inequality above, we used the fact that $H_{v\sim \tau}{v_1} = H_{v\sim \tau}{v_i}$ for all $i\in[b]$, due to $\tau$ being a histogram-type. Thus, for $r=1$, the corresponding exponent in \eqref{eq:Var(X)/E(X)2Bound} is
\begin{align*}&(f(1) -H(\tau)+ (1-R))n \\&\le \left((1-R)-\frac{H(\tau)}b\right)n \\&\le -\eps n.\end{align*}
In the latter case, since $f(b)=0$, the exponent is
$$(-H(\tau)+(1-R)b)n \le -\eps b n.$$
We conclude that \begin{align*}\PR{\CRC^n(R)\text{ does not satisfy }\cP}&\le \Pr(X_{M_\rho}=0) \\&\le q^{-\eps n+o(n)}.\qedhere \end{align*}
\end{proof}

\paragraph{Applying the framework to list-recovery.} In the rest of the paper, we use Theorem \ref{thm:threshold} to compute the threshold rate for $(p,\ell,L)$ list-recovery in several different settings. In order to do that, we set up a few useful definitions.

\begin{definition}[$\beta(p,\ell,L)$ and $\histypes_{p,\ell,L}$]\label{def:betaT}
Given $L \in \N$, $\ell \le L$ and $p \in [0,1)$, let $\histypes_{p,\ell,L}$ denote the set of all histogram-types in $T_{p,\ell,L}$. Let
$$\beta(p,\ell,L) = \max_{\tau \in \histypes_{p,\ell,L}}H(\tau).$$
\end{definition}

Theorem~\ref{thm:threshold} allows us to characterize the threshold rate for $(p,\ell,L)$-list recovery in terms of $\beta(p,\ell,L)$:
\begin{corollary}\label{cor:beta}
 Fix $L \in \N$, $\ell \le L$ and $p \in [0,1)$. The threshold rate for $(p,\ell,L)$ list-recovery is 
\[R^* = 1-\frac{\beta(p,\ell,L)}{L}.\]
 \end{corollary}
 \begin{proof}
By Corollary \ref{cor:ListRecoveryIsAnAmazingProperty} and Lemma \ref{lem:SymmetricHasHighEntropy}, $$\beta(p,\ell,L) = \max_{\tau\in T_{p,\ell,L}} H(\tau).$$
The claim now follows from Corollary \ref{cor:ListRecoveryIsAnAmazingProperty} and Theorem \ref{thm:threshold}.
 \end{proof}

Finally, we introduce the following notation, which will be used for the rest of the paper.
\begin{definition}[$P_\ell(\cdot)$ and $D_{d,\ell,L}$]\label{def:Pell}
Fix $\ell\le L$. Given a vector $v\in \Sigma^L$ let 
$$P_\ell(v) = \min_{A\in \binom{\Sigma}\ell}\left|\left\{i\in[L]\mid v_i\notin A\right\}\right|$$  We use the notation $D_{d,\ell,L} = \{v\in \Sigma^L\mid P_\ell(v) = d\}$.
\end{definition}

	
	\section{Bounds on the threshold rate for noisy list-recovery}\label{sec:ThresholdComputation}

The main result in this section is an estimate of $\beta(p,\ell,L)$ (Proposition~\ref{prop:betaComputation} below), which leads to an estimate on the threshold rate for list-recovery (Corollary~\ref{cor:thresh}).  This estimate is very sharp when $\frac{q \log L}{L}$ is small; in subsequent sections we will derive estimates which are more precise for certain parameter regimes.

Before coming to these bounds, we begin with a few useful lemmas that bound $|D_{d,\ell,L}|$ and characterize $\histypes_{p,\ell,L}$.

\begin{lemma}\label{lem:DBound}
Let $r =  1-\frac \ell q$ and $s = \frac dL$.  Suppose that $s < r$. Then,

\begin{align*}&\binom{q}{rq}\binom{L}{sL}\binom{(1-s)L}{\;\underbrace{\frac{(1-s)L}{(1-r)q},\ldots\frac{(1-s)L}{(1-r)q}}_\ell\;}\\&\cdot\binom{sL}{\;\underbrace{\frac{sL}{rq},\ldots\frac{sL}{rq}}_{q-\ell}\;}\le 
|D_{d,\ell,L}| \\&\le 
\binom{q}{rq} \left(\sum_{i=0}^{sL}\binom{L}{i}  \left((1-r)q\right)^{L-i} \left(rq\right)^{i}\right).\end{align*}
\end{lemma}
\begin{proof}
For the upper bound, note that a vector $v\in D_{d,\ell,L}$ is uniquely determined by the combination of the following:
\begin{itemize}
    \item The set $A\in \binom{\Sigma}{\ell}$, consisting of the $\ell$ most common entries in $v$ (with some arbitrary tie-breaking rule). There are $\binom{q}{rq}$ ways to choose this set.
    \item An assignment of elements from $\Sigma$ to $v_i$ ($i\in [L]$) so that at most $sL$ of the entries do not belong to $A$. There are $    \sum_{i=0}^{sL}\binom{L}{i}\cdot  \left((1-r)q\right)^{L-i}\cdot \left(rq\right)^{i}$ such assignments.
\end{itemize}
The upper bound is obtained by multiplying these two counts.

For the lower bound, let $M$ denote the collection of all vectors $v\in \Sigma^L$ for which there exists a set $A_v\in \binom{\Sigma}\ell$ such that:
\begin{itemize}
    \item Each element of $A_v$ appears in $v$ exactly $\frac{L-d}\ell$ times. 
    \item Each element of $\Sigma\setminus A_v$ appears in $v$ exactly $\frac d{q-\ell}$ times.
\end{itemize}
Since we know that $r > s$, we can say that $L > \frac{dq}{q-\ell} = \frac{dl}{q - \ell} + \frac{d(q-\ell)}{q - \ell}$ and therefore $\frac{L-d}{\ell} > \frac{d}{q - \ell}$. From this we may conclude that $M\subseteq D_{d,\ell,L}$. So, 
\begin{align*}&|D_{d,\ell,L}| \ge |M| \\&\ge \binom{q}{rq}\binom{L}{sL} \binom{(1-s)L}{\;\underbrace{\frac{(1-s)L}{(1-r)q},\ldots,\frac{(1-s)L}{(1-r)q}}_\ell\;} \\&\cdot\binom{sL}{\;\underbrace{\frac{sL}{rq},\ldots,\frac{sL}{rq}}_{q-\ell}\;}. \qedhere \end{align*}

\end{proof}

Using Stirling's approximation, Lemma \ref{lem:DBound} immediately yields the following.
\begin{corollary}\label{cor:logD}
In the setting of Lemma \ref{lem:DBound}, suppose that $s< r$. Then,
$$\log_q |D_{d,\ell,L}| = L(1-\DKL srq) \pm O(q\log L),$$
where the underlying constant is universal.
\end{corollary}

In order to compute $\beta(p,\ell,L)$, we will make use of the following characterization of $\histypes_{p,\ell,L}$ (Definition~\ref{def:betaT}).
Intuitively, this lemma says that a histogram-type $\tau$ is bad for $(p,\ell,L)$-list-recovery if and only if it has many symbols inside the most frequent $\ell$ symbols in expectation.

\begin{lemma}\label{lem:HistTypesCharacterization}
Let $1\le \ell\le q$, $L\in \N$ and $0\le p\le 1$. Let $\tau$ be a distribution over $\Sigma^L$ and suppose that $\tau$ is a histogram-type. Then, $\tau \in \histypes_{p,\ell,L}$ if and only if
\begin{equation}\label{eq:HistTypesCharacterizationInequality}
\Eover{u\sim \tau}{P_\ell(u)} \le pL.    
\end{equation}
\end{lemma}
\begin{proof}
Suppose that $\tau \in \histypes_{p,\ell,L}$. A fortiori, $\tau \in T_{p,\ell,L}$. Let $\rho$ be as in Definition \ref{def:TpellL}. Then, 
$$\Eover{v\sim \tau}{P_\ell(u)} \le \Eover{(u,K)\sim \rho}{\left|\{i\in [L]\mid u_i\notin K \}\right|} \le pL,$$
where the first inequality is due to Condition \ref{enum:rhoMarginalIsTauForT} of Definition \ref{def:TpellL}, and the second inequality follows from \eqref{eq:rhoSmallRadiusInequality}. 

Conversely, suppose that \eqref{eq:HistTypesCharacterizationInequality} holds. Let $\rho$ be the distribution of the pair $(u,K)$, where $u$ is first sampled from $\tau$, and let $K=K(u)\in \binom{\Sigma}\ell$ be the lexicographically minimal set for which $\left|\{i\in [L]\mid u_i\notin K\}\right| = P_\ell(v)$. Clearly, $\rho$ satisfies Condition \ref{enum:rhoMarginalIsTauForT} of Definition \ref{def:TpellL}. 

We next show that it also satisfies Condition \ref{enum:rhoSmallRadiusForT}. Observe that $K(u)=K(\pi(u))$ for every permutation $\pi$ over $[L]$. Hence, since $\tau$ is a histogram property, $$\PRover{(u,K)\sim \rho}{u_i\notin K} = \PRover{(u,K)\sim \rho}{u_1\notin K}$$ for every $i\in [L]$. Consequently, for any $i \in [L]$,
\begin{align*}\PRover{(u,K)\sim \rho}{u_i\notin K} &= \frac{1}{L} \Eover{(u,K)\sim \rho}{|\{i\in [L]\mid u_i\notin K\}|} \\&= \frac 1L  \Eover{v\sim \tau}{P_\ell(v)} \le p.\end{align*}
It follows that $\tau\in T_{p,\ell,L}$. Since $\tau$ is a histogram-type, it also belongs to $\histypes_{p,\ell,L}$.
\end{proof}


Now, we come to our estimate on the threshold rate for $(p,\ell,L)$ list-recovery in the regime where $L \to \infty$ and $q\le o(\frac {\log L}L)$.  We begin with the following proposition, which bounds the quantity $\beta(p,\ell,L)$.

\begin{proposition}\label{prop:betaComputation}
Let $r = 1-\frac \ell q$ and suppose that $p\le r$. Then,
$$\beta(p,\ell,L) = L(1-\DKL prq) \pm O(q\log L).$$
\end{proposition}
\begin{proof}
We first bound $\beta(p,\ell,L)$ from below. Let $\tau$ denote the uniform distribution on $D_{pL,\ell,L}$. This distribution clearly belongs to $\histypes_{p,\ell,L}$, so Corollary \ref{cor:logD} yields
\begin{align*}\beta(p,\ell,L) &\ge H(\tau) \\&\ge L(1-\DKL prq) \pm O(q\log L).\end{align*}

We turn to proving a matching upper bound. Let $\tau\in \histypes_{p,\ell,L}$ such that $\beta(p,\ell,L) = H(\tau)$. Let $v\sim \tau$ and denote the distribution of $P_\ell(v)$ by $\mu$. Note that
\begin{align*}H(\tau) &= H(\mu) + H(v\mid P_\ell(v))\\&\le H(v\mid P_\ell(v)) + \log_q L,\end{align*}
since $P_\ell(v)$ has at most $L$ possible values. Now,
\begin{align}\label{eq:betaBound-BoundOnH}
    H(v\mid P_\ell(v)) &= \EE_{d\sim \mu} \left[H(v\mid P(v)=d)\right]\\&\le \EE_{d\sim\mu} \left[\log_q |D_{d,\ell,L}|\right].
\end{align}
Let $$f(d) = \begin{cases}
L\left(1-\DKL {\frac dL}rq\right) &\text{if } \frac dL < r\\
L &\text {if } \frac dL \ge r.
\end{cases}$$
By \eqref{eq:betaBound-BoundOnH}, Corollary \ref{cor:logD}, and the concavity of $f$,
\begin{align*}
H(v\mid P_\ell(v)) &\le \EE_{d\sim \mu}\left[f(d)\right]\ +O(q\log L)\\&\le f\left(\EE_{d\sim \mu}[d]\right) +O(q\log L).
\end{align*}
By \eqref{eq:rhoSmallRadiusInequality}, we have 
$$\Eover{d\sim \mu}{d} \le p.$$
Since $f$ is non-decreasing, it follows that
\begin{align*}&H(v\mid P_\ell(v))\le f(p) +O(q\log L)\\&=  L(1-\DKL prq) + O(q\log L),\end{align*}
establishing the upper bound.
\end{proof}

\begin{corollary}\label{cor:thresh} The threshold rate for $(p,\ell,L)$ list-recovery of a random code is
$$R^* = \begin{cases}
\DKL prq \pm O\left(\frac{q\log L}{L}\right) & \text{if }p<r\\
0 & \text{if }p\ge r,
\end{cases}$$
where $r = 1-\frac \ell q$.
\end{corollary}
\begin{remark}\label{rem:listreccap}
In order to better illustrate the threshold rate computed in Corollary~\ref{cor:thresh}, one can verify the identity
\begin{align*}
    &\DKL{p}{1-\ell/q}{q} \\&= 1 - p\log_q\inparen{\frac{q-\ell}{p}} - (1-p)\log_q\inparen{\frac{\ell}{1-p}}.
\end{align*}
Substituting $\ell=1$, we find $\DKL{p}{1-1/q}{q} = 1-h_q(p)$, agreeing with the list decoding capacity theorem. For larger $\ell$, this expression agrees with the \emph{list-recovery capacity theorem}, as stated in e.g.~\cite{resch2020list}. 
\end{remark}
\begin{proof}
The case $p<r$ is immediate from Corollary \ref{cor:beta} and Proposition \ref{prop:betaComputation}.

For $p\ge r$, let $\tau$ denote the uniform distribution on $\Sigma^L$. We claim that $\tau\in \histypes_{p,\ell,L}$. Since $\tau$ is clearly a histogram-type, it suffices to show that it satisfies \eqref{eq:HistTypesCharacterizationInequality}. Fix some arbitrary set $A\in \binom \Sigma\ell$. Then,
\begin{align*}\Eover{u\sim \tau}{P_\ell(u)} &\le \Eover{u\sim \tau}{|\{i\in[L]\mid u_i\notin A\}|}\\&\le L\cdot \left(1-\frac \ell q\right) = Lr \le Lp,\end{align*}
and the claim follows. 

Corollary \ref{cor:beta} now yields
$$R^* = 1-\frac{\beta(p,\ell,L)}L \le 1-\frac{H(\tau)}{L} = 0.\qedhere$$
\end{proof}

	\section{Zero-error list-recovery and perfect hashing codes}\label{sec:hash}
	In this section we analyze the threshold rate for zero-error list-recovery (that is, when $p=0$), 
and give a more precise version of Corollary \ref{cor:thresh} in this setting. 
We use this to compute the threshold rate for a random code be to a perfect hash code, which is the same as being $(0,q-1,q)$ list-recoverable. 

\begin{lemma}\label{lem:zero} Let $p^* = |D_{0,\ell,L}|/q^L$. The threshold rate for $(0,\ell,L)$ list-recovery of a random code is
$$R^* =
\frac{-\log_q(p^*)}{L}
$$
\end{lemma}

\begin{proof}
 Let $\mu$ be the uniform distribution on $D_{0,\ell,L}$. Note that $\mu\in T_{0,\ell,L}$, so \begin{equation}\label{eq:betap=0UpperBound}
     \beta(0,\ell,L)\ge H(\mu)=\log_q(|D_{0,\ell,L}|) = \log_q(p^*) + L.
 \end{equation}
  On the other hand, due to \eqref{eq:rhoSmallRadiusInequality}, every distribution in $T_{0,\ell,L}$ is supported in $D_{0,\ell,L}$. Consequently, \eqref{eq:betap=0UpperBound} is in fact an equality (since $\log_q(|D_{0,\ell,L}|)$ is simply the entropy of the uniform distribution, which is the maximal entropy distribution on $D_{0,\ell,L}$). It now follows from Corollary \ref{cor:beta} that
 $$R^* = 1 - \frac{\beta(0,\ell,L)}{L}.\qedhere$$
\end{proof}

\begin{corollary}\label{cor:hash}
The threshold rate for $(0,q-1,q)$ list-recovery of a random code is \[R^* = \frac{1}{q}\log_q\left(\frac{1}{1 - q!/q^q}\right) \]
\end{corollary}

\begin{proof}
Due to Lemma \ref{lem:zero} it suffices to show that
    \[|D_{0,q-1,q}| = q^q - q!.\]
Indeed, $|D_{0,q-1,q}|$ counts the number of vectors of length $q$, over an alphabet of size $q$ such that not all $q$ distinct letters show up in the vector. This number is exactly all possible vectors after removing the permutations of $(1,2,\dots,q)$, namely, $q^q-q!$.
\end{proof}
	
	\section{List of two decoding of random and random linear codes} \label{sec:list-of-2}
	
In this section, we study the list-of-2 decodability of two random ensembles of codes. In detail, we precisely compute the threshold rate for $(p,3)$-list-decoding for random codes and for random \emph{linear} codes. Denote by $\cP$ the monotone increasing property of \emph{not} being $(p,3)$-list-decodable. Note that we cannot immediately apply Corollary \ref{cor:thresh}, as the error term of $O\inparen{\frac{q\log L}{L}}$ is not negligible in this regime. We specialize to the case of $q=2$, and recall our convention that $\log$ denotes the base-2 logarithm. Recall from the introduction that whenever $p < 1/4$ there exist $(p,3)$-list-decodable codes with positive rate, but whenever $p > 1/4$ the only $(p,3)$-list-decodable codes are of bounded size, independent of $n$. 

Our main result of this section is a demonstration that the list-of-2 decoding threshold rate for random \emph{linear} codes is in fact greater than the corresponding threshold rate for random codes. This result demonstrates that our techniques are precise enough to allow us to sharply delineate between different natural ensembles of codes.

In the following, $\CRLC^n(R)$ denotes a random linear code of block length $n$ and rate $R$. We define the threshold rate for random linear codes, denoted $ \TRLC^n(\cP)$, in a manner analogous to the definition for random codes. It is,
\[
  \sup\inset{R \in [0,1]\suchthat \PR{\CRLC^n(R) \text{ satisfies }\cP} \leq \tfrac{1}{2}}
\]

if such an $R$ exists, and $0$ otherwise.

\begin{theorem} \label{thm:list-of-2}
    Let $p \in (0,1/4)$. 
    \begin{enumerate}
        \item \label{item:list-of-2-I} The threshold rate for $(p,3)$-list-decoding for random codes satisfies 
        \[
            \lim_{n \to \infty}\TRC^n(\cP) = 1 - \frac{1+h(3p)+3p\log3}{3} .
        \]
        \item \label{item:list-of-2-II} The threshold rate for $(p,3)$-list-decoding for random \emph{linear} codes satisfies
        \[
            \lim_{n \to \infty}\TRLC^n(\cP) = 1-\frac{h(3p) + 3p \log3}{2} .
        \]
    \end{enumerate}
\end{theorem}

Note that the threshold rate for random linear codes is greater than the threshold rate for random codes, uniformly over $p \in (0,1/4)$. See Figure~\ref{fig:rate-thresh}. 

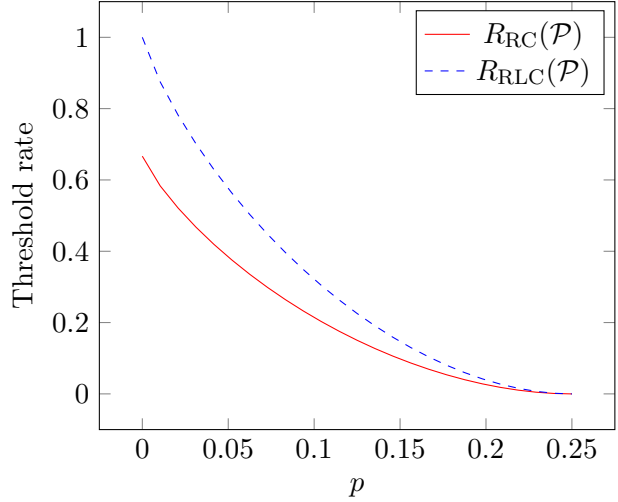
\begin{figure} 
	\centering
	\begin{tikzpicture}
	\begin{axis} [
	    title=Plots of the threshold rates,
		xlabel={$p$},
		ylabel={Threshold rate},
		xticklabels={-0.05, 0, 0.05, 0.1, 0.15, 0.2, 0.25}
		]
		
		\addplot[red,domain=0:0.25] {1-(1/3)*( (3*x)*ln(1/(3*x))/ln(2) + (1-3*x)*ln(1/(1-3*x))/ln(2) + 1 + 3*x*ln(3)/ln(2) )};
		\addlegendentry{$\TRC(\cP)$}
		
		\addplot[blue,domain=0:0.25,dashed]  {1 - 0.5*( ((3*x)*ln(1/(3*x))/ln(2) + (1-3*x)*ln(1/(1-3*x))/ln(2)) + 3*x*ln(3)/ln(2) )};
		\addlegendentry{$\TRLC(\cP)$}
		
		\end{axis}
		\end{tikzpicture}
	\caption{The threshold rate $R_{RC}$ (red) for $(p,3)$-list-decodability of random codes, and the threshold rate $R_{RLC}$ (blue, dashed) for $(p,3)$-list-decodability of random \emph{linear} codes. Note that, uniformly over $p$, random linear codes have the greater threshold rate. } 
	\label{fig:rate-thresh}
\end{figure}

The proof of Theorem~\ref{thm:list-of-2} is split into two parts. Part~\ref{item:list-of-2-I}, which computes the threshold rate for random codes, is proved in the remainder of this subsection. The proof of Part~\ref{item:list-of-2-II} concerning the threshold rate for random linear codes is deferred to Appendix~\ref{app:list-of-2-rlc}, as it requires certain tools and techniques developed in other works (\cite{MRRSW,resch2020list}) that we have not yet introduced in this paper.

\begin{proof} [Proof of Theorem~\ref{thm:list-of-2}, Part~\ref{item:list-of-2-I}]
By Corollary~\ref{cor:beta}, to compute the threshold rate, it suffices to show
\[
    \beta(p,1,3) = \max_{\tau \in \histypes_{p,1,3}}H(\mu) = 1+h_2(3p)+3p\log3.
\]
Once we have done this, Corollary~\ref{cor:beta} will tell us that
\[
    1 - \frac{\beta(p,1,3)}{3} = 1 - \frac{1+h(3p)+3p\log3}{3} ,
\]
as required. We will also use the characterization of $\histypes_{p,1,3}$ provided by Lemma~\ref{lem:HistTypesCharacterization}: namely, $\tau \in \histypes_{p,1,3}$ if and only if $\Eover{u \sim \tau}{P_1(u)}\leq 3p$. We first prove an upper bound on $\beta(p,1,3)$, followed by a matching lower bound. 

\paragraph{Upper bound on $\beta(p,1,3)$} Let $\tau \in \histypes_{p,1,3}$. Let $A = \{(0,0,0),(1,1,1)\}$ and let $B = \F_2^3\setminus A$. Define $x := \mu(B) = \PRover{v \sim \tau}{v \in B}$, and note that $\Eover{v \sim \tau}{P_1(v)} = x$. Indeed, $P_1(v) = 1$ if $v \in B$ and $P_1(v)=0$ if $v \in A$. Thus, we deduce $x \leq 3p$. Since $|A|=2$ and $|B|=6$, we may upper bound 
\begin{align*}
    H(\tau) &\leq x\log\frac{6}{x} + (1-x)\log\frac{2}{1-x} \\&= 1+h(x)+x\log3 . 
\end{align*}
As $x \leq 3p < \frac{3}{4}$, one can see that the quantity $1+h(x)+x\log3$ is increasing in $x$: indeed, defining $g:[0,1/4] \to \R$ by $x \mapsto 1+h(x)+x\log3$, note that $g'(x) = \log\inparen{\frac{3(1-x)}{x}}$ is positive if and only if $3(1-x)<x$, which is satisfied precisely when $x \in (0,\tfrac{3}{4})$. Therefore, $H(\tau) \leq 1+h(3p)+3p\log3$. Thus,  
\begin{align*}
    \beta(p,1,3) &= \max_{\tau \in \histypes_{p,1,3}}H(\tau) \\&\leq 1+h(3p)+3p\log3 .
\end{align*}

\paragraph{Lower bound on $\beta(p,1,3)$} Now, define the distribution $\tau^*$ by $\tau^*(v) = \frac{1-3p}{2}$ for $v \in A$ and $\tau^*(v) = \frac{3p}{6} = \frac{p}{2}$ for $v \in B$. Reasoning as above, $\Eover{v \sim \tau^*}{P_1(v)} = |B| \cdot \frac{3p}{6} = 3p$, and moreover it is clear that $\tau(\pi(v))=\tau^*(v)$ for all permutations $\pi:[3]\to[3]$ and $v \in \F_2^3$. Thus, $\tau^* \in \histypes_{p,1,3}$. Furthermore 
\begin{align*}
    H(\tau^*) &= |B|\cdot \frac{p}{2}\log\frac{2}{p} + |A|\cdot\frac{1-3p}{2}\log\frac{2}{1-3p} \\&= 1+h(3p)+3p\log3 ,
\end{align*}
so
\begin{align*}
    \beta(p,1,3) &= \max_{\tau \in \histypes_{p,1,3}}H(\tau) \geq H(\tau^*) \\&= 1+h(3p)+3p\log3 . \qedhere 
\end{align*}
\end{proof}

	\section{Computing the threshold rate for list-recovery efficiently}\label{sec:efficient}

In the previous sections, we gave precise analytical expressions for the threshold rate for list-recovery in certain parameter regimes.  However, there are some regimes where these bounds aren't precise.  In this section, we consider the question of computing the threshold rate $R^*$ algorithmically, given $p,\ell$ and $L$.  We use tools from the study of entropy-maximizing distributions to develop a simple binary-search-based procedure to pinpoint $R^*$ up to arbitrarily small additive error.

We begin with a lemma that shows that we can compute the cardinality $|D_{d,\ell,L}|$ efficiently; we will use this as a subroutine in our final algorithm.

\begin{lemma}\label{lem:DComputation}
Given $0\le d\le L$ and $1\le \ell\le q$, the cardinality $\left|D_{d,\ell,L}\right|$ can be computed in time $$O\left((L+1)^q + \poly(q,L)\right).$$
\end{lemma}
\begin{proof}
Given a vector $v\in \Sigma^L$, define its \emph{histogram} $\eta_v : \Sigma\to \Z_{\ge 0}$ by $$\eta_v(\sigma) =  |\{i\in [L]\mid v_i = \sigma\}|.$$
Note that the number of vectors in $\Sigma^L$ with a given histogram $\eta$ is given by the multinomial $$\binom{L}{\left(\eta(\sigma)\right)_{\sigma \in \Sigma}}.$$
Also, a vector $v$ belongs to $D_{d,\ell,L}$ if and only if 
\begin{equation}\label{eq:DCriterionByEta}
L - \max_{A\in \binom\Sigma\ell} \sum_{\sigma\in A}\eta(\sigma) = d.
\end{equation}
Thus, 
$$\left|D_{d,\ell,L}\right| = \sum_{\eta} \binom{L}{\left(\eta(\sigma)\right)_{\sigma \in \Sigma}},$$
where the sum is over all distributions $\eta$ satisfying \eqref{eq:DCriterionByEta}. In particular, $\left|D_{d,\ell,L}\right|$ can be computed by going over all $(L+1)^q$ functions $\eta:\Sigma\to \{0,\ldots, L\}$ and summing the terms corresponding to those functions that satisfy \eqref{eq:DCriterionByEta}. 
\end{proof}

We recall the following standard facts from the theory of entropy-maximizing distributions. 
\begin{lemma}\cite[Sec.\ 3]{WJ08}\label{lem:ExponentialFamilies}
Let $\Omega$ be a finite nonempty set, $f:\Omega\to \R$ and $t\in \R$. Let $S_t$ denote the set of all distributions $\tau$ over $\Omega$ such that $\Eover{\omega\sim \tau}{f(\omega)} = t$. Let $$F(t) = \max_{\tau \in S_t}H(\tau).$$
Then
$$ F(t) = \inf_{\alpha\in \R} \left[\log_q \left(\sum_{\omega \in \Omega} q^{\alpha\cdot f(\omega)}\right) - \alpha t\right].$$
Furthermore:
\begin{enumerate}
    \item\label{enum:ExponentialFamiliesSymmetry} If $\tau$ is the entropy maximizing distribution, then $\tau(\omega) = \tau(\omega')$ for every $\omega,\omega'\in \Omega$ such that $f(\omega)=f(\omega')$.
    \item\label{enum:ExponentialFamiliesMonotone} Let $t^* = \Eover{\omega\sim \Uni(\Omega)}{f(\omega)}$. Then, $F(t^*) = \log |\Omega|$, and $F(t)$ is nondecreasing (resp.\ nonincreasing) in the range $t < t^*$ (resp. $t > t^*$).
    \item\label{enum:ExponentialFamiliesConvex} The function $$\log_q \left(\sum_{\omega \in \Omega} q^{\alpha\cdot f(\omega)}\right) - \alpha t$$ is convex in $\alpha$. 
\end{enumerate}
\end{lemma}

\begin{lemma}\label{lem:betaExplicit}
Let $\ell\le q$, $L\in \N$ and $0< p\le 1$, and let $t^* = q^{-L}\cdot\sum_{d=0}^L{d\cdot\left|D_{d,\ell,L}\right|}$. Then, in the case when $p< t^*$,

\[\beta(p,\ell,L) = \inf_{\alpha \in \R} \left[\log_q \left(\sum_{d=0}^L \left|D_{d,\ell,L}\right|\cdot q^{\alpha d}\right) - \alpha p L\right]\]
and when $p \ge t^*$,
\[\beta(p,\ell,L) = L.\]

\end{lemma}

\begin{remark} \label{rem:zero-rate-regime}
In general, $\frac{t^*}L$ is slightly smaller than $1-\frac \ell q$. Thus, Lemma \ref{lem:betaExplicit} lemma extends the range in which the threshold is $0$ from $\left[1-\frac \ell q,1\right]$ (Corollary \ref{cor:thresh}) to $[t^*,1]$.
\end{remark}

\begin{proof}
Suppose that $p \ge t^*$. Let $\tau$ be the uniform distribution over $\Sigma^L$. Recalling the Definition of $P_\ell(u)$ from Definition~\ref{def:Pell}, we have $\Eover{u\sim\tau}{P_\ell(u)} = t^*L$. Consequently,
Lemma \ref{lem:HistTypesCharacterization} implies that $\tau\in \histypes_{p,\ell,L}$. Hence, $\beta(p,\ell,L)\ge H(\tau) = L$, and this bound is clearly tight since no distribution over $\Sigma^L$ has entropy larger than $L$.

We proceed, assuming that $p < t^*$. Let $S_{t,\ell,L}$ be the set of all distributions $\tau$ over $\Sigma^L$ such that $\Eover{u\sim \tau}{P_\ell(u)} = tL$. By Lemma \ref{lem:ExponentialFamilies}(\ref{enum:ExponentialFamiliesSymmetry}), the distribution $\tau$ that has maximal entropy in $S_{t,\ell,L}$, satisfies $\tau(u)=\tau(v)$ whenever $P_\ell(u)=P_\ell(v)$. In particular, $\tau$ is invariant to coordinate permutations, so it is a histogram-type. Therefore, using Lemma~\ref{lem:HistTypesCharacterization} and the definition of $\beta(p,\ell,L)$,
$$\beta(p,\ell,L) = \max_{t\le p} \max_{\tau\in S_t}H(\tau).$$
By Lemma  \ref{lem:ExponentialFamilies}(\ref{enum:ExponentialFamiliesMonotone}), this expression is maximized by $t=p$, namely
$$\beta(p,\ell,L) = \max_{\tau\in S_p}H(\tau).$$
Consequently, by the main part of Lemma \ref{lem:ExponentialFamilies},
\begin{align*}&\beta(p,\ell,L) \\&= \min_{\alpha\in \R}\left[\log_{q}\left(\sum_{u\in \Sigma^L}q^{\alpha P_\ell(u)}\right) - \alpha pL\right] \\&= \min_{\alpha\in \R}\left[\log_{q}\left(\sum_{d=0}^L \left|D_{d,\ell,L}\right|\cdot q^{\alpha d}\right) - \alpha pL\right].\qedhere\end{align*}
\end{proof}

\begin{theorem}\label{cor:algorithm}
There is an algorithm, that, given $p$, $\ell$, $L$ and $\eps > 0$, computes the threshold-rate for $(p,\ell,L)$-list-recovery, within an additive error of $\eps$, in time $O\left((L+1)^q + \poly(q,L,\log \frac 1\eps,\beta(p)\right)$,
where
$$\beta(p) = \begin{cases}
\log\frac 1p&\text{ if } p>0\\
1&\text{ if } p=0.
\end{cases}
$$
\end{theorem}
\begin{proof}Consider the following algorithm:
\begin{enumerate}
    \item Compute the coefficients $|D_{d,\ell,L}|$ for $0\le d\le L$.
    \item If $p=0$, return the expression for $R^*$ given in Lemma \ref{lem:zero}.
    \item If $p\ge t^*$ (where $t^*$ is as in Lemma \ref{lem:betaExplicit}.), return $0$.
    \item If $p < t^*$, use the bisection method (i.e., binary search) to approximate the minimum $M$ of the function $$g(\alpha) := \log_q \left(\sum_{d=0}^L \left|D_{d,\ell,L}\right|\cdot q^{\alpha d}\right) - \alpha p L$$ to within an additive error of $\eps L$. Return $1-\frac ML$.
\end{enumerate}

The correctness of the algorithm follows from Lemma \ref{lem:betaExplicit}. The first step can be completed in time $O\left((L+1)^q + \poly(q,L)\right)$ due to Lemma \ref{lem:DComputation}. 

To analyze the last step, we first note that $g(\alpha)$ is convex, due to Lemma \ref{lem:ExponentialFamilies}(\ref{enum:ExponentialFamiliesConvex}), and thus has a unique minimum. In addition, it is $L$-Lipshitz. Finally, as we show in Claim~\ref{claim:straightforward} below,  the minimizing $\alpha$ lies in the range $\left[-\left(L+\log_q \frac 1p\right),0\right]$. It follows that this step requires at most $O(\log L + \log \log \frac 1p+\log \frac 1\eps)$ bisection iterations.

It remains to prove the following claim.
\begin{claim}\label{claim:straightforward}
\[\argmin_\alpha g(\alpha) \in \left[ - \left(L + \log_q\frac{1}{p}\right), 0\right].\]
\end{claim}
\begin{proof}
We first compute
$$\frac{dg}{d\alpha} = \frac{\left(\sum_{d=0}^L \left|D_{d,\ell,L}\right|\cdot d\cdot q^{\alpha d}\right)}{\left(\sum_{d=0}^L \left|D_{d,\ell,L}\right|\cdot q^{\alpha d}\right)} - pL.$$

The derivative at $0$ is positive due to our assumption that $p < t^*$, which, along with the fact that $g$ is convex, implies that the minimizer is less than zero. It is left to prove that the derivative at $\alpha_0 := -(L+\log_q \frac 1p)$ is negative, which will imply that the minimizer is at least $-(L+\log_q \frac 1p)$. 

To see this, let $Z = \sum_{d=0}^L \left|D_{d,\ell,L}\right|\cdot q^{\alpha_0 d}$. Then,
\begin{align*}
\frac{dg}{d\alpha}_{|\alpha=\alpha_0} &= 
 \frac{\left(\sum_{d=0}^L \left|D_{d,\ell,L}\right|\cdot d\cdot q^{\alpha_0 d}\right)}{Z} - pL \\&\le \frac{L\cdot \left(\sum_{d=1}^L \left|D_{d,\ell,L}\right|\cdot \cdot q^{\alpha_0 d}\right)}{Z} - pL
 \end{align*}
so it suffices to show that 
$$\left(\sum_{d=1}^L \left|D_{d,\ell,L}\right|\cdot \cdot q^{\alpha_0 d}\right) \le pZ.$$
The left-hand side is  $Z - |D_{0,\ell,L}|$, so the above is equivalent to 
$$\frac{|D_{0,\ell,L}|}Z \ge 1-p.$$
Now,
$$
Z \le |D_{0,\ell,L}| + q^L \cdot q^{\alpha_0} = |D_{0,\ell,L}| + p
$$
so indeed,
$$
\frac{|D_{0,\ell,L}|}Z \ge \frac{|D_{0,\ell,L}|}{|D_{0,\ell,L}|+p} 
$$
which is at least $1-p$ since $|D_{0,\ell,L}| \ge 1$.
\end{proof}
This completes the proof of the claim, and thus of Corollary~\ref{cor:algorithm}.
\end{proof}

	\section{Acknowledgements}
	We would like to thank Ray Li for helpful conversations.
	\bibliographystyle{alpha}
	\bibliography{refs}
	
	\appendix
	
	\section{Example of a non-symmetric property}\label{app:example}
	
In this section, we give an example of a property that is row-symmetric (aka, is defined by a collection of bad types), but is not symmetric. Then we will see that Theorem~\ref{thm:threshold} is \em not \em satisfied for this property.

Here is one such property $\mathcal{P} = \{P_n\}_{n \in \mathbb{N}}$.  Let $p \in (0,1/2)$. Say that $C \subset \{0,1\}^n$ satisfies $\mathcal{P}$ if there are three codewords $c^{(1)}, c^{(2)}, c^{(3)} \in C$ so that for some point $z \in \{0,1\}^n$, 
$d(c^{(1)}, z), d(c^{(2)},z) \leq \frac{pn}{2}$, and
 $d(c^{(3)},z) \leq pn$.

\def\toy{\text{toy}}

That is, this property is similar to list-of-two decoding, except that two of the codewords must be closer to the central point than the third must be.  This property is not symmetric, and Theorem~\ref{thm:threshold} does not hold for it.  However, in order to see an example that is as clean as possible, let's instead consider a ``toy'' version of this property.

Define the property $\mathcal{P}^{\toy}$ as follows.  Let $T = \{\tau_1, \tau_2, \tau_3\}$, where
\begin{align*} &\tau_1((1,0,0)) = \tau_1((0,1,0)) = \frac{p}{2} \\& \tau_1((0,0,1)) = p \\& \tau_1((0,0,0)) = 1 - 2p \end{align*}
\begin{align*} &\tau_2((1,0,0)) = \tau_2((0,0,1)) = \frac{p}{2} \\& \tau_2((0,1,0)) = p \\& \tau_2((0,0,0)) = 1 - 2p \end{align*}
\begin{align*} &\tau_3((0,0,1)) = \tau_3((0,1,0)) = \frac{p}{2} \\& \tau_3((1,0,0)) = p \\& \tau_3((0,0,0)) = 1 - 2p \end{align*}
(and $\tau_i((1,1,1))=0$ for all $i$).
That is, $\tau_1$, $\tau_2$, and $\tau_3$ correspond to the matrices $B$ so that two columns have exactly a $p/2$ fraction of $1$'s, one column has exactly a $p$ fraction of $1$'s, and the support of the columns are disjoint.  (For simplicity, assume that $pn/2$ is an integer.)  Then $\mathcal{P}^{\toy}$ is defined by the inclusion of at least one type in $T$.  

The toy version $\mathcal{P}^{\toy}$ is in spirit the same as the property $\mathcal{P}$ above. The differences are that (a) we are fixing the vector $z$ to be the all-zero vector, (b) we demand that the distance be exactly $p$ or $p/2$, rather than at most $p$ or $p/2$, and (c) we are not considering configurations where the three vectors have overlapping ones.  The computations are similar for both $\mathcal{P}^{\toy}$ and $\mathcal{P}$.  We mention both, but analyze $\mathcal{P}^{\toy}$, since $\mathcal{P}$ is perhaps more natural, while $\mathcal{P}^{\toy}$ is simpler to analyze.

First, observe that $\mathcal{P}^{\toy}$ is not symmetric.  In particular, the distribution $\frac{1}{3}(\tau_1 + \tau_2 + \tau_3)$, which corresponds to a matrix with $2p/3$ fraction of $1$'s in each column, is not contained in $T$.

Theorem~\ref{thm:threshold} would predict that the threshold rate for $\mathcal{P}^{toy}$ would be 
\begin{align*} R^{\text{theorem}} &= 1 - \max_{\tau \in T}H(\tau)/3 \\&= 1 - \frac{1}{3}\biggl( p \log\left( \frac{2}{p^2}\right) \\& \biggr. \quad+ \biggl. (1 - 2p) \log\left( \frac{1}{1 - 2p} \right) \biggr). \end{align*}
However, the actual threshold rate is larger than this.  To see this, observe that the probability that $C$ satisfies $\mathcal{P}^{\toy}$ is at most the probability that $C$ contains the type $\tau^\dagger$ on $\{0,1\}^2$ given by
\begin{align*} &\tau^\dagger((1,0)) = \tau^\dagger((0,1)) = \frac{p}{2} \\& \tau^\dagger((0,0)) = 1 - p \\& \tau^\dagger((1,1)) = 0. \end{align*}
Indeed, suppose that $C$ contains $\tau_i$ for some $i=1,2,3$.  Then in particular it contains two codewords of weight $p/2$, meaning that it contains $\tau^\dagger$.  Therefore,
\[ R^* \geq R^\dagger, \]
where $R^\dagger$ is the rate threshold for the property corresponding to $\{\tau^\dagger\}$.  But the rate threshold for $\tau^\dagger$ is 
\begin{align*} R^\dagger &= 1 - \frac{H(\tau^\dagger)}{2} \\&= 1 - \frac{1}{2} \left( p \log\left( \frac{2}{p} \right) + (1 - p) \log\left( \frac{1}{1-p}\right) \right). \end{align*}
Figure~\ref{fig:example} plots these two values, and it is clear that $R^\dagger$ (and hence $R^*$) is larger than $R^{\text{theorem}}$ for $p \leq 0.3$.

\begin{figure}
	\centering
	\begin{tikzpicture}
	\begin{axis} [
	    title={Theorem~\ref{thm:threshold} does not hold for $\mathcal{P}^{\toy}$},
		xlabel={$p$},
		xticklabels={-0.05, 0,  0.1,  0.2, 0.3, 0.4, 0.5}
		]
		
		\addplot[red,domain=0:0.5] {1-(1/3)*( (x)*ln(2/(x^2))/ln(2) + (1-2*x)*ln(1/(1-2*x))/ln(2)  )};
		\addlegendentry{$R^{\text{theorem}}$}
		
		\addplot[blue,domain=0:0.5,dashed]  {1 - 0.5*( ((x)*ln(2/(x))/ln(2) + (1-x)*ln(1/(1-x))/ln(2)) )};
		\addlegendentry{$R^\dagger$}
		
		\end{axis}
		\end{tikzpicture}
	\caption{The value $R^{\text{theorem}}$ that Theorem~\ref{thm:threshold} would predict for $\mathcal{P}^{\toy}$, and the value $R^\dagger$ that is a lower bound on the actual value of $R^*$, for $p \in (0,0.5)$. } 
	\label{fig:example}
\end{figure}
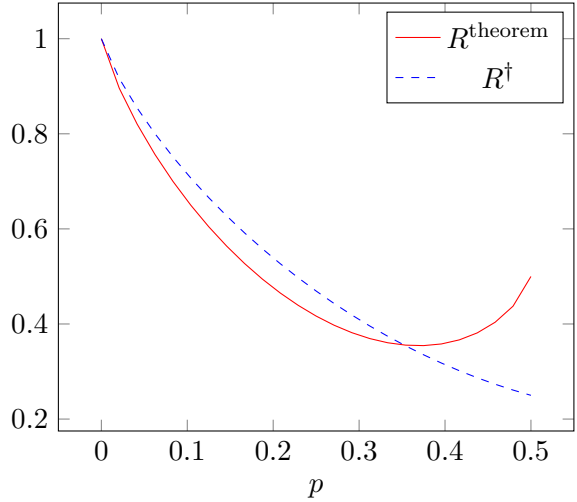
	
	\section{Proofs of claims from Sections \ref{sec:Preliminaries} and \ref{sec:ThersholdCharacterization}}\label{app:ProofsForCharacterization}
	
We restate and prove the deferred lemma.

\GoodApproximationForNonRecoverability*
\begin{proof}
Let $\tau_0\in T_{p,\ell,L}$, and let $\rho_0$ be a distribution over $\Sigma^L\times \binom\Sigma\ell$ that satisfies Conditions \ref{enum:rhoMarginalIsTauForT} and \ref{enum:rhoSmallRadius} with regard to $\tau_0$. Let $\tau_1$ denote the atomic distribution over $\Sigma^L$ that always returns the all-$\sigma$ vector, for some arbitrary $\sigma\in \Sigma$. Let $\rho_1$ be the atomic distribution over $\Sigma^L\times \binom\Sigma\ell$ that always returns $(\sigma,W)$, where $W\in \binom\Sigma\ell$ is some arbitrary set containing $\sigma$. 

Fix some $\eps>0$ and let $\tau_\eps$ and $\rho_\eps$ denote the mixture distributions $(1-\eps)\tau_0+\eps\tau_1$ and $(1-\eps)\rho_0+\eps\rho_1$, respectively. Let $\{(u^i,K^i)\}_{i=1}^n$ denote a sequence of $n$ independent samples from $\rho_\eps$. 

Sample $i$ uniformly from $[n]$. Let $\hat \tau$ denote the distribution of $(u^i)$ and let $\hat \rho$ be the distribution of $(u^i,K^i)$. It is straightforward to verify that $\hat \rho$ satisfies Conditions \ref{enum:rhoMarginalIsTauForT} and \ref{enum:rhoDivisibility} with regard to $\hat \tau$. If $p=0$ then Condition \ref{enum:rhoSmallRadiusForT} also clearly holds. 

Suppose that $p > 0$. We claim that, with some positive probability, Condition \ref{enum:rhoSmallRadiusForT} is satisfied as well, and consequently, $\hat\tau \in T_{p,\ell,L}^n$. 

Let $X_j$ denote the number of coordinates $i$ for which $u^i_j\notin K^i$. Note that the CDF of $X_j$ is bounded from below by that of a $\mathrm{Binomial}(n,(1-\eps)p)$ variable. Consequently, by Hoeffding's bound, \begin{align*}
&\PRover{\hat{\rho}}{\exists j\in L \text{ s.t. } \PRover{(u,K)\sim \hat \rho}{u_j\in K} > p} \\&\le
L\cdot\PRover{\hat{\rho}}{\PRover{(u,K)\sim \hat \rho}{u_j\in K} > p} \\ &= L\cdot\PR{X_j > pn} \\
&\le L\cdot e^{-2(p\eps)^2n}.
\end{align*}
Hence, taking $\eps > \sqrt{\frac{\log_q L}{2p^2n}}$, Condition \ref{enum:rhoSmallRadiusForT} is satisfied with positive probability. Therefore, there exists $\hat \tau\in T_{p,\ell,L}$ with $d_\infty (\tau,\hat\tau) \le \eps = O\left(\sqrt{\frac{\log_q L}{p^2n}}\right)$.
\end{proof}
	
	\section{Threshold rate for list-of-two decoding of random linear codes} \label{app:list-of-2-rlc}
	In this section, we prove Part~\ref{item:list-of-2-II} of Theorem~\ref{thm:list-of-2}. For convenience, the statement is provided below. 

\begin{theorem} [Theorem~\ref{thm:list-of-2}, Part~\ref{item:list-of-2-II}, Restated] \label{thm:list-of-2-restated}
    Let $p \in (0,1/4)$. The threshold rate for $(p,3)$-list-decoding for random linear codes satisfies 
    \[
        \lim_{n \to \infty} \TRLC^n(\cP) = 1 - \frac{h(3p)+3p\log 3}{2} \ .
    \]
\end{theorem}

Following \cite{MRRSW}, we can again characterize the threshold rate for random linear codes in terms of types.\footnote{Again, we warn the reader that in that work monotone\emph{-decreasing} properties were considered, whereas here we consider monotone-increasing properties.} However, the characterization is now more complicated; in particular, we need to account for so-called \emph{implied} distributions. We now introduce the necessary machinery, specialized to the $q=2$ case. 

For a distribution $\tau \sim \F_2^b$, consider the span of the support of $\tau$. The dimension of this subspace is denoted by $\dim(\tau)$.

\begin{definition} [Implied Distribution] \label{def:implied-dist}
    Let $b \in \N$, let $\tau$ be a distribution over $\F_2^b$ and let $A\in \F_2^{m \times \ell}$ be a rank $m$ matrix for some $m \leq \ell$. The distribution of $Au \in \F_2^m$, where $u$ is sampled according to $\tau$, is said to be $\tau$\emph{-implied}. We denote the set of $\tau$-implied types by $\cI_\tau$. 
\end{definition}

We briefly motivate the above definition; for further details, the reader is referred to \cite{MRRSW}, specifically Section~2. Suppose a linear code $C$ contains a matrix $M \in \F_2^{n \times b}$ of type $\tau$. Then, it also contains any matrix $M' \in \F_2^{n \times m}$ whose columns lie in the column span of $M$. That is, for any full-rank matrix $A \in \F_2^{m \times b}$, it contains the matrix $MA^T$. Moreover, note that if $M$ has type $\tau$, then the type of $M' = MA^T$ is $\tau'$ as in Definition~\ref{def:implied-dist}. 

We now quote the result from \cite{MRRSW} that we require. 

\begin{theorem} [Theorem~2.8 of \cite{MRRSW}] \label{thm:mrrsw-rlc}
    Fix $b \in \N$. Let $(P_n)_{n \in \N}$ be a sequence of monotone-increasing properties such that $T_n \subseteq \cT_b^n$ is a minimal set for $P_n$ for all $n \in \N$. Then
    \[
        \TRLC^n(P_n) = 1-\max_{\tau \in T_n}\min_{\tau' \in \cI_\tau}\frac{H(\tau')}{\dim(\tau')} \pm o_{n \to \infty}(1).
    \]
\end{theorem}

With these tools, we may prove Theorem~\ref{thm:list-of-2-restated}.

\begin{proof} [Proof of Theorem~\ref{thm:list-of-2-restated}]
    Specializing to the specific property $\cP = (P_n)_{n \in \N}$ of $(p,3)$-list-decodability, recall that $T_n$ should consist of all types $\tau \in \cT_3^n$ for which there exists a distribution $\rho$ over $\F_2^3 \times \F_2$\footnote{Technically, to be completely consistent with Section~\ref{sec:Preliminaries}, we should consider distributions over $\F_2^3 \times \binom{\F_2}{1}$, but we just identify elements with the corresponding singleton sets in the natural way.} such that, given $(u,z) \sim \rho$, the following holds:
    \begin{enumerate}
        \item\label{cond:u-is-tau} The distribution of $u$ is $\tau$.
        \item\label{cond:p-bounded} $\PR{u_j \neq z}\leq p$ for every $1 \leq j \leq 3$.
        \item\label{cond:integrality} $n\rho((u,z)) \in \N$ for every $u \in \F_2^3$ and $z \in \F_2$.
        \item\label{cond:distinct-cols} Any matrix $M \in M_\tau$ has distinct columns, i.e., $M \in \dis{(\F_2)}{n \times 3}$. Stated differently, $\PR{u_i \neq u_j}>0$ for any $1\leq i < j \leq 3$.
    \end{enumerate}
    Define 
    \[
        \bar\gamma(p,1,3) := \lim_{n \to \infty}\max_{\tau \in T_n}\min_{\tau' \in \cI_\tau}\frac{H(\tau')}{\dim(\tau')};
    \]
    Theorem~\ref{thm:mrrsw-rlc} says that $\lim_{n \to \infty}\TRLC^n(P_n) = 1-\bar\gamma(p,1,3)$. We will show $\bar\gamma(p,1,3) = \frac{h(3p)+3p\log 3}{2}$. We begin by proving an upper bound on $\bar\gamma(p,1,3)$.
    
    \paragraph{Upper bound on $\bar\gamma(p,1,3)$.} Let $\tau \in T_n$ and let $\rho \sim \F_2^{3}\times \F_2$ be the promised distribution satisfying Conditions~\ref{cond:u-is-tau}--\ref{cond:distinct-cols}. Condition~\ref{cond:p-bounded} implies that 
    \begin{align} \label{eq:list-of-two-avg-rad-first}
        \frac{1}{3}\sum_{j=1}^3 \PRover{(u,z) \sim \rho}{u_j\neq z} \leq p .
    \end{align}
    Let $\mathrm{MAJ}(u)$ denote the majority element of the vector $u$,\footnote{In symbols, $\mathrm{MAJ}(u) = \argmax\{b \in \F_2:|\{j \in [3]:u_j=b\}|\}$.} and note that if $z = \mathrm{MAJ}(u)$ the left-hand side of \eqref{eq:list-of-two-avg-rad-first} can only decrease. Hence,
	\begin{align} \label{eq:list-of-two-maj}
		\frac{1}{3} \sum_{i=1}^3 \PRover{u \sim \tau}{u_i \neq \mathrm{MAJ}(u)} \leq p .
	\end{align}
	Now, as in Section~\ref{sec:list-of-2} let\footnote{In this proof, we denote a vector by the corresponding string for readability.} $A = \{000,111\}$ and $B = \F_2^3 \setminus A$. Then, defining $x = \tau(B)$, Condition~\eqref{eq:list-of-two-maj} becomes
	\begin{align} \label{eq:list-of-two-x-ub}
		x \leq 3p .
	\end{align}
	Now, consider the implied type $\tau^* \in \cI_\tau$ defined by the linear map $(a,b,c) \mapsto (a+b,a+c)$. Note that the kernel of this map is $\{000,111\}$, and so $\tau^*(00) = 1-x$. Hence, $\tau^*(10) + \tau^*(01) + \tau^*(11) = x$. We may therefore bound the entropy from above as follows:
	\begin{align*}
		H(\tau^*) &= \tau^*(00)\cdot \log\tfrac{1}{\tau^*(00)} + \tau^*(01)\cdot \log\tfrac{1}{\tau^*(01)} \\& \quad + \tau^*(10)\cdot \log\tfrac{1}{\tau^*(10)} + \tau^*(11)\cdot \log\tfrac{1}{\tau^*(11)} \\
		&\leq (1-x) \log\tfrac{1}{1-x} + x\log\tfrac{3}{x} \\&= h(x) + x \log 3 \ .
	\end{align*}
	The above inequality uses the concavity of the function $y \mapsto y \log \tfrac{1}{y}$. Note that in the range $[0,3/4)$, the function $x \mapsto h(x) + x\log 3$ is increasing: clearly $h_2(0) + 0\cdot \log_23 = 0$, and moreover the derivative of $h(x) + x\log_23$ with respect to $x$ is $\log\inparen{\tfrac{3(1-x)}{x}}$, which is positive assuming $\frac{3(1-x)}{x} > 1$, which rearranges to $x < 3/4$. Hence, as $x \leq 3p$ and $p < 1/4$, we conclude 
	\[
		H(\tau^*) \leq h(3p) + 3p\log3 \ .
	\]
	Now, we claim that $\dim(\tau^*) = 2$. Let $U = \mathrm{span}(\supp(\tau))$. If $\dim(\tau^*) \leq 1$, then $\dim(U) \leq 2$ and $111 \in U$ (recall that the kernel of $(a,b,c)\mapsto (a+b,b+c)$ is $\{000,111\}$). This implies 
	\begin{align*}
		U \in &\{\{000\}, \{000,111\}, \{000,111,001,110\}, \\& \{000,111,010,101\}, \{000,111,100,011\}\} \ .
	\end{align*}
	In any of the above cases, we find that $\tau$ contradicts Condition~\ref{cond:distinct-cols}. For example, if $U = \{000,111,001,110\}$, then $\PR{u_i \neq u_j}=0$. 
	
	As $\tau \in T_n$ was arbitrary, we conclude that 
	\begin{align*}
	    \bar\gamma(p,1,3) &= \lim_{n \to \infty}\max_{\tau \in T_n}\min_{\tau' \in \cI_\tau}\frac{H(\tau')}{\dim(\tau')} \\&\leq \frac{h(3p)+3p\log3}{2},
	\end{align*}
	as desired.
	
	\paragraph{Lower bound on $\bar\gamma(p,1,3)$.} For any $n \in \N$, we define a type $\tau_{n,0} \in T_n$ and show that 
	\[
	    \min_{\tau' \in \cI_{\tau_{n,0}}}\frac{H(\tau')}{\dim(\tau')} \geq \frac{h(3p)+3p\log3}{2} - o_{n \to \infty}(1);
	\]
	taking limits will then show $\bar\gamma(p,1,3) \geq \frac{h(3p)+3p\log3}{2}$, as desired. 
	
	Define $\rho_n \sim \F_2^{3 \times 1}$ by assigning probability mass $\frac{1}{n}\floor{\frac{p}{2}\cdot n}$ to each of the elements $(001,0),(010,0),(100,0)$, $(011,1),(101,1),(110,1) \in \F_2^3 \times \F_2$, and splitting the remaining probability mass as evenly as possible between the $(000,0)$ and $(111,1)$, being sure to obey the condition that $n\cdot \rho_n(000,0)$ and $n \cdot \rho_n(111,1)$ are integers. Intuitively, $\rho_n$ is tending to the ``limit'' distribution $\rho$ defined by 
	\begin{align*}
		\rho(000,0)&=\rho(111,1)=\frac{1-3p}{2} ~~ \text{and}\\
		\rho(001,0)&=\rho(010,0)=\rho(100,0)=\rho(011,1)\\&=\rho(101,1)=\rho(110,1) = \frac{p}{2};
	\end{align*}
	however, as we require a type lying in $T_n$, we cannot directly work with $\rho$. Let $\tau_{n,0}$ denote the distribution of $u$ for $(u,z) \sim \rho_n$. It is immediate that Conditions \ref{cond:p-bounded}, \ref{cond:integrality} and \ref{cond:distinct-cols} are satisfied for $\rho_n$ and $\tau_{n,0}$. Also, observe that $\dim(\tau_{n,0}) = 3$, as $\tau_{n,0}$ has full-support for sufficiently large $n$ (this uses the assumption $p \in (0,1/4)$). 
	
	Our plan now is to consider all of the implied types of $\tau_{n,0}$ one-by-one. We will show that each such implied type $\tau_{n,0}'$ satisfies 
	\[
	    \frac{H(\tau_{n,0}')}{\dim(\tau_{n,0}')} \geq \frac{h(3p)+3p\log3}{2} - o_{n\to\infty}(1),
	\]
	and therefore, we will deduce
	\[
	    \min_{\tau_{n,0}' \in \cI_{\tau_{n,0}}}\frac{H(\tau_{n,0}')}{\dim(\tau_{n,0}')} \geq \frac{h(3p)+3p\log3}{2} - o_{n\to\infty}(1).
	\]
	First, one can compute that
	\[
		\frac{H(\tau_{n,0})}{\dim(\tau_{n,0})} \geq \frac{h_2(3p)+1+3p\log_23}{3} - o_{n \to \infty}(1).
	\]
	Next, consider the type $\tau_{n,1}$ implied by the map $(a,b,c) \mapsto (a+b,a+c)$. One can calculate
	\[
		\frac{H(\tau_{n,1})}{\dim(\tau_{n,1})} \geq \frac{h_2(3p) + 3p\log3}{2} - o_{n \to \infty}(1).
	\]
	Next, consider any type $\tau_{n,2}$ implied by a full-rank map $\F_2^3 \to \F_2^2$ with a vector from $B$ in the kernel. We compute
	\[
		\frac{H(\tau_{n,2})}{\dim(\tau_2)} \geq \frac{h(2p)+1}{2} - o_{n\to\infty}(1).
	\]
	Next, consider any type $\tau_{n,3}$ implied by a full-rank map $\F_2^3 \to \F_2$ with $111$ in the kernel. In this case,
	\[
		\frac{H(\tau_{n,3})}{\dim(\tau_{n,3})} \geq h(2p) - o_{n \to \infty}(1). 
	\]
	Finally, consider any type $\tau_{n,4}$ implied by a full-rank map $\F_2^3 \to \F_2$ without $111$ in the kernel. In this case, 
	\[
	    \frac{H(\tau_{n,4})}{\dim(\tau_{n,4})} \geq 1-o_{n\to\infty}(1).
	\]
	This completes the computation of $\frac{H(\tau_n')}{\dim(\tau_n')}$ for each $\tau_n' \in \cI_{\tau_{n,0}}.$
	
	 It is now just a finite check to verify that \[\frac{H(\tau_n')}{\dim(\tau_n')} \geq \frac{h(3p)+\log 3}{2}-o_{n\to\infty}(1).\] For example, to show that 
	\begin{align*}
		\frac{H(\tau_{n,0})}{\dim(\tau_{n,0})} &= \frac{h(3p)+1+3p\log3}{3} - o_{n \to \infty}(1) \\&\geq \frac{h(3p)+\log3}{2} - o_{n \to \infty}(1),
	\end{align*}
    one can reason as follows. First, it clearly suffices to show
	\[
	    \frac{h(3p)+1+3p\log3}{3} \geq \frac{h(3p)+\log3}{2}.
	\]
	Now, note that both the left-hand side and the right-hand side are $0$ at $p = 1/4$, and moreover the derivative of the right-hand side minus the left-hand side is $\frac{1}{6}\inparen{\log\inparen{\frac{3p}{1-3p}} - \log3}$, which is negative if $p<1/4$. We omit the remaining computations, which are completely routine; for a pictorial proof, see Figure~\ref{fig:exp-thresh}.
	\FloatBarrier
	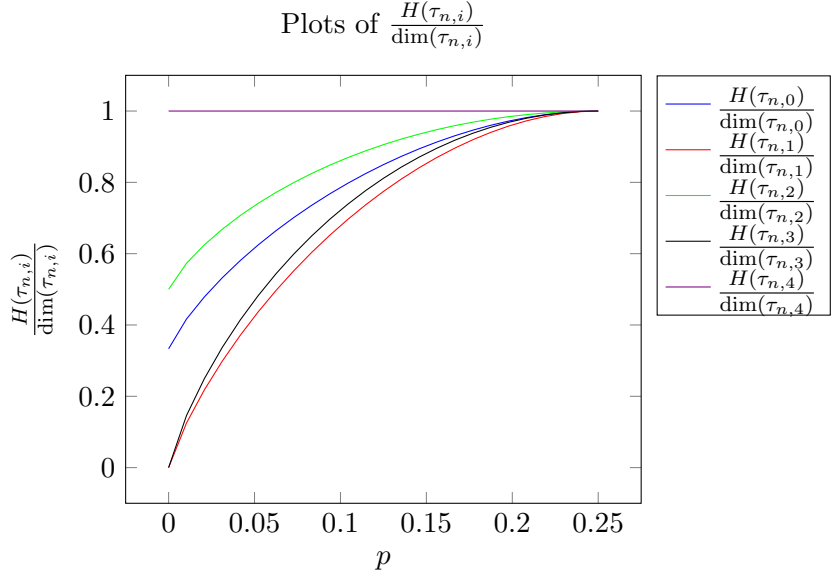
\begin{figure}[h!]
		\centering
		\begin{tikzpicture}[scale= .7]
		\begin{axis} [
		title=Plots of $\frac{H(\tau_{n,i})}{\dim(\tau_{n,i})}$,
		xlabel={$p$},
		ylabel={$\frac{H(\tau_{n,i})}{\dim(\tau_{n,i})}$},
		xticklabels={-0.05, 0, 0.05, 0.1, 0.15, 0.2, 0.25},
		legend pos=outer north east
		]
		\addplot[blue,domain=0:0.25] {(1/3)*( (3*x)*ln(1/(3*x))/ln(2) + (1-3*x)*ln(1/(1-3*x))/ln(2) + 1 + 3*x*ln(3)/ln(2) )};
		\addlegendentry{$\frac{H(\tau_{n,0})}{\dim(\tau_{n,0})}$}
		
		\addplot[red,domain=0:0.25]  {0.5*( ((3*x)*ln(1/(3*x))/ln(2) + (1-3*x)*ln(1/(1-3*x))/ln(2)) + 3*x*ln(3)/ln(2) )};
		\addlegendentry{$\frac{H(\tau_{n,1})}{\dim(\tau_{n,1})}$}
		
		\addplot[green,domain=0:0.25] {0.5*( (2*x)*ln(1/(2*x))/ln(2) + (1-2*x)*ln(1/(1-2*x))/ln(2) + 1 )};
		\addlegendentry{$\frac{H(\tau_{n,2})}{\dim(\tau_{n,2})}$}
		
		\addplot[black,domain=0:0.25] {(2*x)*ln(1/(2*x))/ln(2) + (1-2*x)*ln(1/(1-2*x))/ln(2)};
		\addlegendentry{$\frac{H(\tau_{n,3})}{\dim(\tau_{n,3})}$}
		
		\addplot[violet,domain=0:0.25] {1};
		\addlegendentry{$\frac{H(\tau_{n,4})}{\dim(\tau_{n,4})}$}
		
		\end{axis}
		\end{tikzpicture}
		\caption{Plots of $\frac{H(\tau_{n,i})}{\dim(\tau_{n,i})}$ for each $i \in \{0,1,2,3,4\}$, ignoring $o_{n\to\infty}(1)$ terms. One can see that, uniformly over $p \in (0,1/4)$, the minimum is obtained by $\frac{H(\tau_{n,1})}{\dim(\tau_{n,1})}$.}
		\label{fig:exp-thresh}
	\end{figure}
    \FloatBarrier
	Thus, we conclude that 
	\[
	    \bar\gamma(p,1,3) = \lim_{n \to \infty}\max_{\tau \in T_n}\min_{\tau' \in \cI_\tau} \frac{H(\tau')}{\dim(\tau')} \geq \frac{h(3p)+3p\log3}{2},
	\]
	completing the proof. 
\end{proof}

\end{document}